\title{Generalized Lyapunov conditions for $k$-contraction: analysis and feedback design}

\author{Andreu Cecilia, Samuele Zoboli,  Daniele Astolfi, Ulysse Serres
 and Vincent Andrieu
\thanks{ $^{1}$ A. Cecilia is with 
Universitat Polit\`ecnica de Catalunya, Avinguda Diagonal, 647, 08028 Barcelona, Spain.
 ({\tt \small  
andreu.cecilia@upc.edu}).
}
\thanks{$^2$  S. Zoboli 
 is with 
Universit\'e Paul Sabatier Toulouse 3, LAAS-CNRS, Toulouse, France 
 ({\tt \small  
name.surname@laas.fr}).
}
\thanks{$^3$  D. Astolfi, U. Serres and V. Andrieu are with 
Universit\'e Claude Bernard Lyon 1, CNRS, LAGEPP UMR 5007, 43 boulevard du 11 novembre 1918, F-69100, Villeurbanne, France 
 ({\tt \small  
name.surname@univ-lyon1.fr}).
}
\thanks{The research leading to these results is partially funded by the Spanish Ministry of Universities funded by the European Union - NextGenerationEU (2022UPC-MSC-93823) and ANR via
projects DELICIO (ANR-18-CE40-0010) and ALLIGATOR (ANR-22-CE48-0009-01).}
}

\documentclass[journal,twoside,web,9pt]{ieeecolor}

\usepackage{generic}
\usepackage{graphicx}
\usepackage{epstopdf}
\usepackage{graphics,psfrag,color,theorem}
\usepackage{amsfonts,latexsym}
\usepackage[noadjust]{cite}
\usepackage{bigints}
\usepackage{amsmath,mathtools}
\usepackage{amssymb}
\usepackage{comment}
\usepackage{color}
\usepackage{float}
\usepackage{mathrsfs}
\usepackage[mathscr]{eucal}
\usepackage[utf8]{inputenc}
\usepackage{tabularx,ragged2e,booktabs}
\usepackage{mathrsfs}
\usepackage{enumerate}
\graphicspath{{Figures/}}
\usepackage{hyperref}
\usepackage{varwidth}
\usepackage{algpseudocode}
\usepackage{algorithm}
\usepackage{multicol}
\usepackage{svg}

\definecolor{MyGreen}{RGB}{50,140,80}

\theoremstyle{plain}
\newtheorem{theorem}{Theorem}

\newtheorem{proposition}{Proposition}

\newtheorem{lemma}{Lemma}
\newtheorem{remark}{Remark}
\newtheorem{definition}{Definition}

\newtheorem{assumption*}{Assumption}

\makeatletter
\let\theoremstyle\@undefined                       
\makeatother

\def\NN{{\mathbb N}}    % set of natural number
\def\RR{{\mathbb R}}    % field of real number
\def\CC{{\mathbb C}}    % field of complex number

\renewcommand{\le}{\leqslant}

\renewcommand{\leq}{\leqslant}
\renewcommand{\geq}{\geqslant}

\newcommand{\defeq}{\vcentcolon=}
\newcommand{\eqdef}{=\vcentcolon}

\newcommand{\cA}{\mathcal A}

\newcommand{\cH}{\mathcal H}
\newcommand{\cI}{\mathcal I}

\newcommand{\cS}{\mathcal S}

\newcommand{\cV}{\mathcal V}

\newcommand{\transf}{\mathbf{T}}
\newcommand{\transfh}{\mathbf{T_h}}
\newcommand{\transfv}{\mathbf{T_v}}
\newcommand{\difflow}{\tfrac{\partial \psi}{\partial x}}

\newcommand{\dddt}{{\dfrac{{\rm d}}{{\rm d}t}}}

\newcommand{\diffvar}{v}
\newcommand{\diffvarzero}{v_{0}}

\newcommand{\dr}{{{\rm d}r}}

\DeclareMathOperator{\trace}{tr}
\DeclareMathOperator{\card}{card}

\renewcommand{\Re}{\mathfrak{Re}}

\newcommand{\vol}{V^k}%\text{Vol}}
\newcommand{\cine}{\mathrm{In}}
\newcommand{\Imag}{\mathrm{Im}}

\newcommand{\vertmatnonlin}[2]{\Psi(#1,#2)}

\newcommand\smallmat[1]{\left[\begin{smallmatrix} #1 \end{smallmatrix} \right]}

\newcommand{\pderiv}[2]{\frac{\partial #1}{\partial #2}}

\begin{document}

\maketitle

\begin{abstract}
Recently, the concept of $\mathit{k}$-contraction has been introduced as a promising generalization of contraction for dynamical systems. However, the study of $\mathit{k}$-contraction properties has faced significant challenges due to the reliance on complex mathematical objects called matrix compounds. As a result,  related control design methodologies have yet to appear in the literature. In this paper, we overcome existing limitations and propose new sufficient conditions for $\mathit{k}$-contraction which do not require matrix compounds computation.
Notably, these conditions are also necessary in the linear time-invariant framework. Leveraging on these findings, we propose a feedback design methodology for both the linear and the nonlinear scenarios which can be used to enforce $\mathit{k}$-contractivity properties on the closed-loop dynamics. 
\end{abstract}


\maketitle
\thispagestyle{empty}
\pagestyle{empty}

%%%%%%%%%%%%%%%%%%%%%%%%%%%%%%%%%%%%%%%%%%%%%%%%%%%%%%%%%%%%%%%%%%%%%%%%%%%%%%%

\section{Introduction}
Contraction theory is an emerging topic that has been used in numerous applications, such as observer design \cite{Sanfelice2012}, multi-agent system synchronization \cite{Russo2009,aminzare2014synchronization,giaccagli2023} and controller design \cite{manchester2017control,giaccagli2022sufficient,giaccagli2023lmi}. Nonetheless, many systems cannot present classical contractivity properties, e.g. multi-stable systems or  systems that admit a (non-trivial)
periodic solution. This fact motivated the study of suitable generalizations, such as horizontal contraction \cite[Section VII]{Forni2014} and transversal exponential stability \cite{andrieu2020characterizations}.
Motivated by the results of Muldowney \cite{Muldowney1990},
the recent work \cite{Wu2022} presented the notion of $k$-contraction as the generalization to $k$-dimensional objects of the standard contraction concept for distances.
As such, $k$-contraction includes classical contraction as the special case $k=1$. For $k>1$, this property can be used to analyze the asymptotic behavior of systems that are not contractive in the classical sense. For example, for $2$-contractive time-invariant systems, every bounded solution converges to an equilibrium point (which may not be the same for every solution). 

Existing sufficient conditions for $k$-contraction are given in terms of a particular matrix compound of the Jacobian of the vector field dynamics \cite{Muldowney1990,Wu2022,angeli2023small}. Although these conditions are adequate for system analysis, their application for feedback design is limited. First, matrix compounds rapidly explode in dimension for low values of $k$ and systems of large dimension. This fact drastically increases the computational complexity of potential feedback design algorithms. Second, the use of matrix compounds
hinders the derivation of a tractable matrix inequality problem for feedback design. 
Consequently, a $k$-contractive design methodology has yet to be developed. 

Considering these limitations, this work presents alternative design-oriented conditions for $k$-contraction that do not rely on matrix compounds, but rather on matrix inequalities on the given system dynamics. In particular, we build upon the generalized Lyapunov matrix inequalities studied for instance in \cite{SMITH1986679,Forni2019}. 
By exploiting these novel conditions, we devise a feedback design methodology in both the linear and the nonlinear framework.
In the linear time-invariant framework, our design is based on a new generalization of the notion of stabilizability, which is, then, generalized to the   nonlinear framework.

The remainder of this document is organized as follows. In Section~\ref{sec:preliminaries}, we provide a refined definition of $k$-contraction which strongly focuses on its geometrical interpretation. 
Then, we recover the notion of infinitesimal $k$-contraction, which has been used in \cite{Wu2022b}, and link it to the proposed definition of $k$-contraction. 
Subsequently, we recall matrix compound-based sufficient conditions for $k$-contraction and  discuss their limitations. 
Section~\ref{sec:lin} focuses on linear systems. First, we derive necessary and sufficient conditions for $k$-contraction that do not require matrix compounds, but rather generalized Lyapunov matrix inequalities. Then, based on these results, we propose the notion of $k$-order stabilizability together with a new $k$-contractive feedback design. 
In doing so, we also collect and recall in a unified theorem a series of results 
on inertia theorems (see, e.g. \cite[Lemma 1, Section 3]{smith1979poincare} and \cite[Theorem 2.5]{stykel2002stability})
and we provide new inertia results
on algebraic inequalities of the form $WA^\top+AW-BB^\top\prec 0$ not requiring any controllability assumption (see for instance \cite{wimmer_1976}).
Section~\ref{sec:nonlin} focuses on extending these results to nonlinear systems. Similarly,  we first provide sufficient conditions for $k$-contraction in nonlinear systems and, then, propose a design methodology for $k$-contraction. In section~\ref{sec:asymptotic}, we discuss 
asymptotic behaviors of $k$-contractive dynamics. 
The results are validated in some numerical simulations in Section~\ref{sec:Illustrations}. All the proofs are postponed in Sections~\ref{sec:proof_L}-\ref{sec:proofs_NL}  and in the Appendix.
Finally, some conclusions and future perspectives are drawn in Section~\ref{sec:conclusions}.

\emph{Contribution:} This work builds on the initial results published in the conference paper \cite{zoboli2023lmi}. The main differences between \cite{zoboli2023lmi} and the current paper are summarized in the following points:
\begin{itemize}
\item We refine the necessary and sufficient conditions for $k$-contraction in linear systems presented in \cite{zoboli2023lmi} and provide a general result for state matrices with arbitrary eigenvalues.
\item We present a state-feedback design methodology for linear systems. In doing so, we introduce the notion of $k$-order stabilizability and an associated generalized Lyapunov test. 
\item We extend the proposed design to nonlinear dynamics.
\item We discuss the asymptotic behavior of $k$-contractive systems, with a specific focus on $k\in\{1,2,3,n\}$.
\item We provide multiple numerical examples of $k$-contraction analysis and design, with a particular focus on the multi-stability properties of a grid-connected synchronverter
\cite{natarajan2018almost,Lorenzetti2022}.
\end{itemize}

\emph{Notation:}  $\RR_{\geq0}\defeq[0,\infty)$ and
$\NN \defeq \{0,1,2,\ldots\}$. $|\cdot|$ denotes the standard Euclidean norm.
Given $x\in \RR^n$, $y\in \RR^m$, we set 
$(x,y)\defeq (x^\top , y^\top)^\top$.
We denote $\binom{n}{k} \defeq  \tfrac{n!}{k!(n-k)!}$ as the binomial coefficient, with $n!$
 denoting the factorial of $n\in \NN$.
The inertia of a matrix $P$
is defined by the triplet of integers
$
\cine(P) \defeq (\pi_{-}(P),  \pi_{0}(P), \pi_{+}(P)),
$
where $\pi_{-}({P})$, $\pi_{+}(P)$ and $\pi_{0}(P)$ denote the numbers of eigenvalues of $P$ with negative,
positive and zero real part, respectively, counting multiplicities. 
The cardinality of a set is denoted as $\card(\cdot)$. $A\succ0$ (resp. $A\succeq 0$) denotes $A$ being a positive definite (resp. positive semidefinite) matrix.

\section{Preliminaries on $k$-contraction}\label{sec:preliminaries}

\subsection{Definition of $\mathit{k}$-contraction}

In this work, we consider nonlinear systems of the form
\begin{equation}\label{eqn:original_system}
    \dot{x} = f(x), \qquad x\in \RR^n
\end{equation}
where $f$ is sufficiently smooth.
The flow of $f$ 
is denoted by  $\psi^t$, and $\psi^t(x_0)$ is the trajectory of \eqref{eqn:original_system} at time $t$. By definition, $\psi^0(x_0)=x_0$.
We now formally define the property of $k$-contraction studied in this article. Our 
definition strongly focuses on a geometrical interpretation and it is related to the notion presented in the works 
\cite{Muldowney1990,Wu2022}. Moreover,  when considering objects of dimension $1$ ($k=1$), it
matches the definition of contraction presented in \cite{Andrieu2016,andrieu2020characterizations}.  

In \cite{Andrieu2016,andrieu2020characterizations}, 
$1$-contraction expresses the fact that the length of any parametrized and sufficiently smooth curve exponentially decreases with time. 
To extend such a notion to any positive integer $k\in [1, n]$, we consider a set of sufficiently smooth functions  $\cI_k$ defined on  $[0,1]^k$, namely
\begin{equation}\label{eqn:parametriz_set}
    \!\! \cI_k \defeq \left\{\Phi:[0,1]^k\to\RR^n \mid \text{ $\Phi$ is a 
    smooth immersion} \right\}\!.
\end{equation}
Let $P\in \RR^{n\times n}$ be a positive definite symmetric matrix. 
For each $\Phi$ in $\cI_k$, we define the volume  $\vol(\Phi)$ of $\Phi$ as
\begin{equation}\label{eqn:k_length}
    \vol(\Phi) \defeq \bigintsss_{[0,1]^k} \sqrt{\det\left\{\frac{\partial \Phi}{\partial r}(r)^\top P\frac{\partial \Phi}{\partial r}(r)\right\}}\, \dr\,.
\end{equation} 
Note that, since $f$ in \eqref{eqn:original_system} is sufficiently smooth,  for each $t$ in $\RR_{\geq0}$ 
the corresponding flow $\psi^t$ is also sufficiently smooth. Consequently, for each $\Phi$ in $\cI_k$, $\psi^t\circ\Phi$ belongs to $\cI_k$.

\begin{remark}
When $\Phi$ is injective  and $P$ is the identity matrix, 
the volume $\vol$ defined in 
\eqref{eqn:k_length}  
coincides with the standard Euclidean $k$-volume of the submanifold $\Phi([0,1]^k)\subsetneq\RR^n$.
Note that  $1$-volumes are lengths, $2$-volumes are areas and $3$-volumes are standard volumes.
\end{remark}

\begin{remark}\label{rem:Riemannian_volume}
    The volume definition \eqref{eqn:k_length} can be generalized to the Riemannian 
    framework by substituting the Euclidean metric $P$ with a symmetric positive definite 2-tensor $P:\RR^n\to\RR^{n\times n}$, see~\cite[Lemma 3.2]{Lee_Riemann}. However, in this paper, we will focus on the Euclidean scenario in order to obtain more tractable conditions.  
\end{remark}

From now on, 
we let $k$ be a fixed integer  
between 1 and $n$. We now define $k$-contraction 
for nonlinear systems of the form \eqref{eqn:original_system}.

\begin{definition}[$\boldsymbol k$-contraction]\label{def:k_contraction} 
System \eqref{eqn:original_system}
is said to be $k$-contractive on a forward invariant set $\cS\subseteq \RR^n$ if there exist  real numbers $a,b>0$ such that, 
for every $\Phi \in \cI_k$ satisfying $\Imag(\Phi) \subseteq \cS$, the following holds 
\begin{equation}\label{eqn:k_contraction_continuous}
    \vol(\psi^t\circ \Phi) \leq b\, e^{-a t}\, \vol(\Phi),
    \qquad \forall t\in \RR_{\geq0}.
\end{equation} 
\end{definition}

In other words, a system is $k$-contractive if, for any parametrized $k$-dimensional submanifold of $\RR^n$ from which trajectories are complete, its volume  exponentially shrinks along the system dynamics. 
An intuitive representation of the required volume convergence condition is presented in Fig.~\ref{fig:k_contraction_scheme}.
When $k=1$, \eqref{eqn:k_contraction_continuous} implies that the length of any sufficiently smooth curve exponentially decreases, matching the definition in \cite{andrieu2020characterizations}. Moreover, this definition includes the ones in \cite{Muldowney1990}, and \cite[Section 3.2]{Wu2022}. 
We highlight the following property, which is directly obtained by using the definition of the volume $\vol$.

\begin{lemma}\label{lem:k+1_contract}
    If
     system  \eqref{eqn:original_system}
    is $k$-contractive for an integer $1\le k<n$, then it is also {$(k+1)$-contractive}.
\end{lemma}

\begin{figure}[t]
	\begin{center}
		\includegraphics[width=0.8\linewidth]{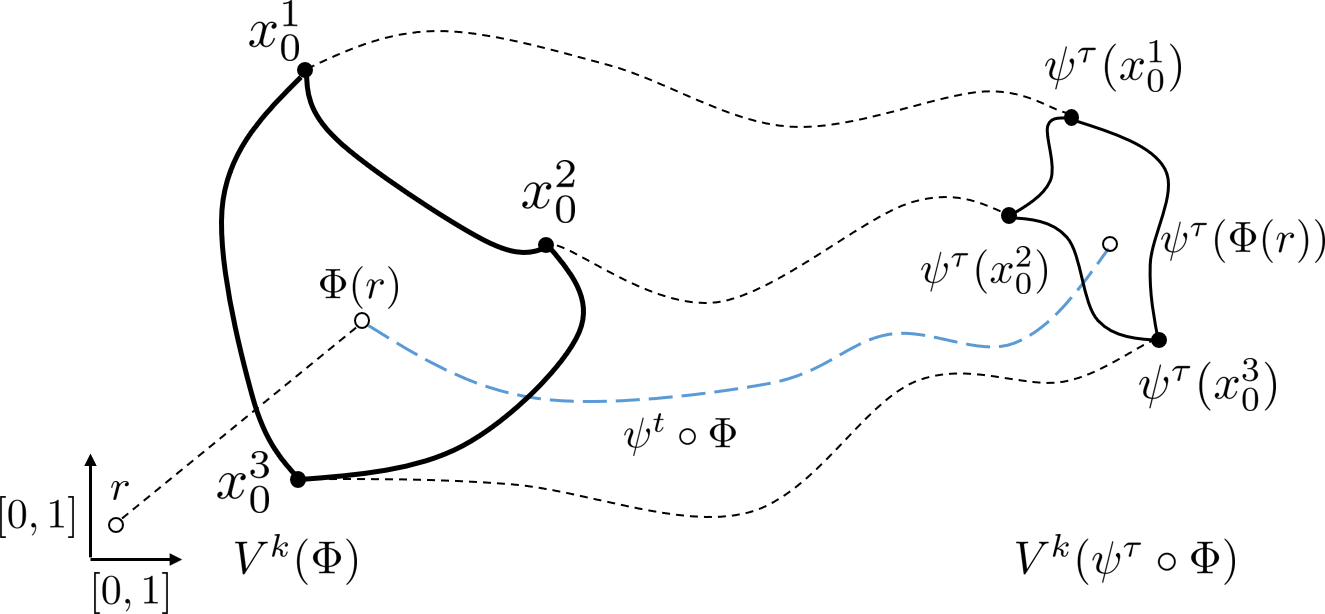}
		\caption{ Flow 
  of a $2$-contractive system. The initial submanifold of initial conditions is described by $\Phi$.  The volume of this submanifold decreases exponentially along the trajectories of the system.
  } \label{fig:k_contraction_scheme}
	\end{center}\vspace{-2em}
\end{figure}

We remark that Definition \ref{def:k_contraction} is invariant under 
globally Lipschitz diffeomorphism on $\cS$, due to the presence of the overshoot term $b$.

\begin{lemma}\label{lem:coordinate_invariance}
Assume that the system \eqref{eqn:original_system} is $k$-contractive on a forward invariant set $\cS\subseteq \RR^n$ for some positive constants $a,b>0$. Moreover, let
$\varphi:\cS\to\cS$ be a diffeomorphism
which satisfies
for some positive constants $\bar c, \underline c>0$
\begin{equation}\label{eqn:unif_diff}
\underline c I \preceq\frac{\partial \varphi}{\partial x}(x)^\top  \frac{\partial \varphi}{\partial x}(x)\preceq  \bar c I, \quad \forall x\in \cS.
\end{equation}
Then, there exists a positive constant $\bar b>0$ such that for every $\Phi \in \cI_k$ satisfying $\Imag(\Phi) \subseteq \cS$, the following holds 
$$
\vol(\varphi\circ \psi^t\circ \Phi)
     \leq \bar b e^{-a t}\, \vol(\varphi\circ\Phi),
$$
\end{lemma}
The proof is postponed to Appendix~\ref{sec:proof_coordinate_inv}.

\subsection{Infinitesimal $\mathit{k}$-contraction}\label{subsec:inifnit_k_contr}
Inspired by classical works on contraction theory \cite{LOHMILLER1998683}, we now provide a result linking the exponential stability properties of the variational system of \eqref{eqn:original_system} to the $k$-contraction property proposed in Definition~\ref{def:k_contraction}. 
We start by recalling the dynamics of the variational system, which describe the evolution of an infinitesimal displacement along the trajectories of the system.
Specifically, the variational system of \eqref{eqn:original_system} 
along the trajectory $\psi^t(x_0)$ is
\begin{equation}\label{eqn:variational_system} 
    \dot{\diffvar} = \dfrac{\partial f}{\partial x}(\psi^t(x_0)) \ \diffvar, \quad  \diffvar\in\RR^n.
\end{equation}
Then, $\difflow^t (x_0) \diffvarzero$ is a trajectory of \eqref{eqn:variational_system} at time $t$ initialized at $\diffvarzero$ at $t=0$. From linearity, we have that $\difflow^t (x_0)$ is the state transition matrix of \eqref{eqn:variational_system}. 
Thus, $ \difflow^t(x_0)\diffvarzero$ depicts the infinitesimal displacement with respect to the solution $\psi^t(x_0)$ induced by the initial condition $ x_0+\diffvarzero$.

We recall that the trajectory $\psi^t(x_0)$ is locally exponentially stable, that is, the trajectory generated from any initial condition close enough to $x_0$ will exponentially converge to $\psi^t(x_0)$, if and only if the variational system \eqref{eqn:variational_system} is exponentially stable \cite[Theorem 3.13]{Khalil2002}. In classical contraction theory \cite{LOHMILLER1998683}, this property is generalized by considering simultaneously all the trajectories in a set. That is, system \eqref{eqn:original_system} is contracting in a forward invariant set $\cS\subseteq\RR^n$ if the variational system \eqref{eqn:variational_system} is exponentially stable for all $x_0\in\cS$. Then, contraction on $\cS$ implies that every solution in $\cS$ converges to the same trajectory \cite{LOHMILLER1998683}, or equivalently, the distance between any pair of trajectories shrinks to zero. In a sense, contraction exemplifies how the local linearization along trajectories can be used to derive global incremental properties of the original system.
In this section, we generalize this idea by considering $k$-contraction properties on the variational system \eqref{eqn:variational_system}.

Pick any $x_0\in\RR^n$ and consider any $k$ initial conditions $\diffvarzero^1,\dots,\diffvarzero^k$ of the variational system \eqref{eqn:variational_system}. We define the following matrix
\begin{equation}\label{eqn:inf_initial_cond}
    \vertmatnonlin{t}{x_0} \defeq \Big[\begin{array}{ccc}\difflow^t (x_0)\diffvarzero^1 & \cdots & \difflow^t(x_0)\diffvarzero^k \end{array}\Big]\in \RR^{n\times k}.
\end{equation}
Note that  $\vertmatnonlin{0}{x_0} = \frac{\partial \Phi_{\rm loc}}{\partial r}(r)$, where $\Phi_{\rm loc}$ is an immersion parameterized by the variable $r\in[0,1]^k$ and whose image is an infinitesimal $k$-order parallelotope with vertices at $x_0$ and $\diffvarzero^i + x_0$, namely, 
$$
\Phi_{\rm loc}(r) = \sum_{i=1}^k r_i(\diffvarzero^i + x_0) + \left(1 -\sum_{i=1}^k r_i\right) x_0,
$$
with $r_i\in[0, 1]$ for $i\in\{1,\dots,k\}$ being the $i$-th component of $r$. For $k=1$,  $\Phi_{\rm loc}(r)$ defines a straight line between $x_0$ and $x_0+v_0^1$. 
The volume of the infinitesimal parallelotope can be computed by means of the multiplicative compound,
  which is defined as follows.

  \begin{figure}[t]
	\begin{center}
		\includegraphics[width=0.9\linewidth]{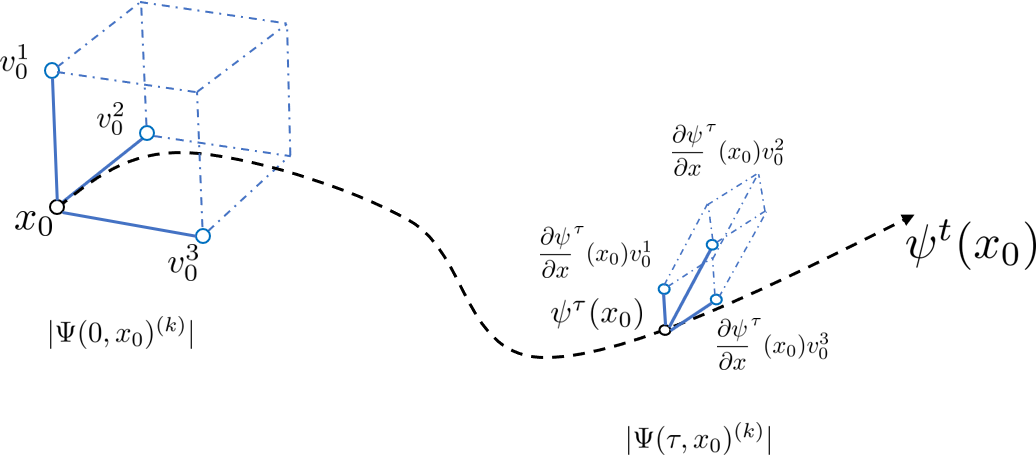}
		\caption{Flow  
  of an   infinitesimally   
  $3$-contractive system.
  } \label{fig:inf_k_contraction}
	\end{center}\vspace{-2em}
\end{figure}

  \begin{definition}[Multiplicative Compound \cite{bar2023compound}]\label{def:compound_definition}
    Consider a matrix $Q\in\RR^{n\times m}$ and select an integer $k\in[1,\min\{n,m\}]$. Moreover, define a minor of order $k$ of the matrix $Q$ as the determinant of some $k\times k$ submatrix of $Q$. The $k$-th multiplicative compound of $Q$, denoted as $Q^{(k)}$, is the  $\binom{n}{k}\times \binom{m}{k}$ matrix including all the minors of order $k$ of $Q$ in a lexicographic order. 
\end{definition}

As an example, consider a $3\times3$ matrix $Q$ with entries $q_{i j}$ for $i,j=1,\dots,3$. The $2^{nd}$ multiplicative compound $Q^{(2)}$ is
$$
Q^{(2)} = \begin{bmatrix} 
\det\left(\begin{smallmatrix}q_{11} & q_{12}\\q_{21} & q_{22}\end{smallmatrix}\right) & \det\left(\begin{smallmatrix}q_{11} & q_{13}\\q_{21} & q_{23}\end{smallmatrix}\right) & \det\left(\begin{smallmatrix}q_{12} & q_{13}\\q_{22} & q_{23}\end{smallmatrix}\right) \\
\det\left(\begin{smallmatrix}q_{11} & q_{12}\\q_{31} & q_{32}\end{smallmatrix}\right) & \det\left(\begin{smallmatrix}q_{11} & q_{13}\\q_{31} & q_{33}\end{smallmatrix}\right) & \det\left(\begin{smallmatrix}q_{12} & q_{13}\\q_{32} & q_{33}\end{smallmatrix}\right) \\
\det\left(\begin{smallmatrix}q_{21} & q_{22}\\q_{31} & q_{32}\end{smallmatrix}\right) & \det\left(\begin{smallmatrix}q_{21} & q_{23}\\q_{31} & q_{33}\end{smallmatrix}\right) & \det\left(\begin{smallmatrix}q_{22} & q_{23}\\q_{32} & q_{33}\end{smallmatrix}\right)
\end{bmatrix}.
$$
Note that for a $Q\in \RR^{n\times n}$, $Q^{(n)}=\det(Q)$ and $Q^{(1)}=Q$.

Thanks to the previous definition, 
by considering $P=I$, we can compute 
the volume of $\Phi_{\rm loc}$ and $\difflow^t\circ \Phi_{\rm loc}$ as follows \cite{Wu2022}
  $$
\begin{aligned}
        \vol(\Phi_{\rm loc}) &= |\vertmatnonlin{0}{x_0}^{(k)}|, \quad
        \vol(\difflow^t\circ \Phi_{\rm loc}) = |\vertmatnonlin{t}{x_0}^{(k)}|.
\end{aligned}
  $$
 The second equality is a consequence of the linearity of the dynamics of $\vertmatnonlin{t}{x_0}^{(k)}$ and we postpone further details at the beginning of Appendix~\ref{subsec:proof_infin_contr}. Given these notions,
  we have the next definition.
  
\begin{definition}[Infinitesimal $\boldsymbol k$-contraction]\label{def:inf_k_contraction}
System \eqref{eqn:original_system}
is said to be infinitesimally $k$-contractive on a forward invariant set $\cS\subseteq \RR^n$ if there exist  real numbers $a,b>0$ such that
    \begin{equation}\label{eqn:compound_stability}
        \left|\vertmatnonlin{t}{x_0}^{(k)}\right| \leq b e^{-a t}\left|\vertmatnonlin{0}{x_0}^{(k)}\right|, 
    \end{equation}
    for all $(t,x_0)\in\RR_{\geq0}\times\cS$ and any $\vertmatnonlin{0}{x_0}$ defined as in \eqref{eqn:inf_initial_cond}.
 \end{definition}

 Roughly speaking, the bound in \eqref{eqn:compound_stability} implies that the volume of any infinitesimal parallelotope connected to the trajectory $\psi^t(x_0)$ and generated by the vectors $ \difflow^t (x_0)\diffvarzero^1, \cdots, \difflow^t(x_0)\diffvarzero^k$ exponentially shrinks to zero. 
 A depiction of this property is presented in Fig.~\ref{fig:inf_k_contraction}. 
Notice that, for the case $k=1$, Definition~\ref{def:inf_k_contraction} boils down to \eqref{eqn:variational_system} being exponentially stable for all $x_0\in\cS$, which is a sufficient condition for the classical notion of contraction \cite{LOHMILLER1998683}. In the next proposition, we link the notion of infinitesimal $k$-contraction to $k$-contraction.

\begin{theorem}\label{thm:inf_k_contr}
Suppose  system \eqref{eqn:original_system} is 
infinitesimally $k$-contractive on a 
forward invariant set $\cS$. 
Then, it is also 
    {$k$-contractive} on $\cS$ according to Definition~\ref{def:k_contraction}.
\end{theorem}
The proof is postponed to  
Appendix~\ref{subsec:proof_infin_contr}.
We highlight that 
in \cite{Andrieu2016} it has also been shown that
for $k=1$, that is, standard contraction, 
the converse is also true, namely
contraction of the system
implies infinitesimally contraction. For $k>1$
the problem is currently open.

\subsection{Sufficient conditions based on additive matrix compounds}\label{sec:matrix_compounds}
Sufficient conditions for $k$-contraction were originally given in the seminal work by Muldowney  \cite{Muldowney1990} and were recently re-proposed  in the works \cite{Wu2022,Angeli2022}. These conditions strongly depend on the use of the additive matrix compound, which is defined below.

\begin{definition}[Additive Compound \cite{bar2023compound}]
    Consider a matrix $Q\in\RR^{n\times n}$ and select an integer $k\in[1,n]$. The $k$-th additive compound of $Q$ is the $\binom{n}{k}\times \binom{n}{k}$ matrix defined as 
$$
Q^{[k]} \defeq 
\left.\frac{\rm d}{{\rm d}\epsilon} \right|_{\epsilon=0} 
(I+\epsilon Q)^{(k)}.
$$ 
\end{definition}

The additive compound can be explicitly computed
in terms of the entries of $Q$.
For example, for $Q\in\RR^{n\times n}$ we have $Q^{[n]}=\trace(Q)$ and $Q^{[1]}=Q$. More details on this operation can be found in \cite{fiedler2008special}.

Bearing this definition in mind, we now reframe the sufficient condition for $k$-contraction presented in \cite{Muldowney1990,Wu2022} in the framework of this paper, 
namely, we view them through the lenses of
Definition~\ref{def:k_contraction}.

\begin{theorem}\label{thm:demidovich}
Let $\cS\subseteq\RR^n$  be a forward invariant set and suppose  
there exist a real number  $\eta>0$  and a symmetric positive definite matrix 
$Q\in\RR^{\binom{n}{k}\times \binom{n}{k}}$
such that
\begin{equation}\label{eqn:compound_LMI}
    Q\left(\tfrac{\partial f}{\partial x}(x)^{[k]}\right)
    + \left(\tfrac{\partial f}{\partial x}(x)^{[k]}\right)^{\!\!\top} Q
    \preceq -\eta I, \quad \forall x\in \cS.
\end{equation}
Then, system \eqref{eqn:original_system} is infinitesimally $k$-contractive on $\cS$ (therefore, $k$-contractive on $\cS$ according to Definition~\ref{def:k_contraction}). 
\end{theorem}

The proof is postponed to 
Appendix~\ref{subsec:proof_demid}. The extension of Theorem~\ref{thm:demidovich}  to time-varying systems can be found in \cite{zoboli2023lmi}.

\begin{remark}
Inequality \eqref{eqn:compound_LMI} is equivalent to the condition in \cite[Theorem 9]{Wu2022} using the logarithmic norm induced by the weighted $\ell_2$ norm (e.g. \cite[Equation 2.56]{bullo2022contraction}). However, 
in our statement, 
the set 
$\cS$ is allowed to be non-convex.
Furthermore, when $k=1$, 
we  recover the well-known Demidovich conditions
(see \cite{davydov2022non})
and
 the proof in 
 \cite{andrieu2020characterizations}
 for contraction of lengths
 in the context of Euclidean metrics. 
 Nonetheless, it is worth noting that,
 similarly to the case of standard 
 contraction (see, e.g. \cite[Section 3]{bullo2022contraction}),
 also   $k$-contraction can  be studied through different logarithmic norms 
 \cite[Section 3.1]{Wu2022}.
\end{remark}

\begin{remark}
    Theorem~\ref{thm:demidovich} can be generalized to the case of Riemannian volumes, see Remark~\ref{rem:Riemannian_volume}.  
    However, we omit these results to ease the reading of the document. It should be remarked that such generalization also expands on point $IV$ in \cite[Proposition 2.5]{SIMPSONPORCO201474}, 
    since we consider volume objects of dimension lower than $n$.
\end{remark}

\subsection{Limitations of matrix compound-based conditions}\label{sec:limitations}
Although Theorem~\ref{thm:demidovich} provides a  
 suitable 
condition for system analysis, we claim that the presence of matrix compounds  
hinders 
the process of 
devising 
$k$-contractive feedback designs.
Indeed, consider a linear control system of the form
\begin{equation}\label{eqn:linear_system_nonautonomous}
\dot x = Ax + Bu,
\quad x\in\RR^n, \quad u\in\RR^m,
\end{equation}
where $u$ is the control input.
Assume we want to design a state-feedback controller of the form $u = -Kx$,
with $K\in\RR^{n\times m}$,
such that the closed-loop system is $k$-contractive. Then, Theorem~\ref{thm:demidovich} reduces to designing $K$ such that condition \eqref{eqn:compound_LMI} is satisfied for the closed-loop system, 
namely,  
$$
Q\left((A-BK)^{[k]}\right)
+\left((A-BK)^{[k]}\right)^\top Q  \preceq -\eta I.
$$
However, this is a non-convex matrix inequality, 
 due to the strong coupling between the matrices $B,K$ imposed by the additive matrix compound. 
Consequently, even for a simple linear case, a design methodology for the gain $K$ cannot be straightforwardly derived.   
 Another motivation for deriving alternative conditions is related to the computational complexity of the matrix inequality \eqref{eqn:compound_LMI}. 
Although  for $k= n$ interesting properties may arise (see below Lemma~\ref{lem:n-contraction}), 
most useful asymptotic behaviors (from a control perspective) require small values of $k$, specifically $k\in\{1,2,3\}$ (see \cite{Wu2022} and Section~\ref{sec:asymptotic}).
Notice that matrix compounds rapidly grow in dimension when the order of the system is large and the $k$ is low, since they involve matrices of dimensions $\binom{n}{k}\times \binom{n}{k}$. 
Consequently, as highlighted in previous works \cite{dalin2022verifying}, compound-based conditions often explode in size.  

With this in mind,  
 a consistent portion of the following sections is 
dedicated the presentation of alternative design-oriented $k$-contraction conditions that do not require matrix compound computation.
These conditions will be fundamental  in the derivation of control laws guaranteeing  $k$-contractivity of the closed-loop.

\section{$k$-contraction  for linear systems}\label{sec:lin}

We start our analysis by focusing on the linear scenario. This will provide fundamental intuitions on the notion of $k$-contraction that will be instrumental in the 
 subsequent 
analysis of nonlinear 
 dynamics.

\subsection{Generalized Lyapunov necessary and sufficient conditions} \label{sub:nec_suff_cond}

Consider a linear system of the form
\begin{equation}\label{eqn:linear_system}
    \dot x = Ax, 
    \qquad x\in \RR^n.
\end{equation}
We now provide 
a set of sufficient and necessary conditions guaranteeing that \eqref{eqn:linear_system} is $k$-contractive according to Definition~\ref{def:k_contraction}. This result is based on the following two facts:
\begin{itemize}
    \item a necessary and sufficient condition for 
    system \eqref{eqn:linear_system} to be $k$-contractive
    is that the sum of the real part of 
    any combination of $k$ eigenvalues of $A$ is negative, 
    see Lemma~\ref{lem:spectral-k-properties} below 
    in Section~\ref{sec:proof_L};
    \item the generalized Lyapunov matrix inequality 
    (see, e.g \cite{SMITH1986679,smith1979poincare})
    $$
    PA + A^\top P \prec 2\mu P
    $$
    admits a symmetric solution $P$ of inertia $\cine(P)=(p, 0 , n-p)$ if and only if $A$ has $p$ eigenvalues with real part larger than $\mu$
and $n-p$ eigenvalues with real part smaller than $\mu$,
    see below Lemma~\ref{lem:inertia_shifting} in Section~\ref{sec:proof_L}.
\end{itemize}
Combining the two properties above, 
we state now the following main result.

\begin{theorem}\label{thm:k_contraction_LMI_cLTI}
System \eqref{eqn:linear_system} is $k$-contractive if and only if
there exist:
\begin{itemize}
    \item a positive integer
$\ell \in \NN$ satisfying 
$1\leq \ell \leq k $, 
\item 
$\ell$
real numbers $\mu_i\in \RR$, with 
$i\in\{0, \ldots, \ell-1\}$, 
\item 
$\ell$ positive integers $d_i\in\NN$,
with 
$i\in \{0, \ldots, \ell-1\}$,
satisfying
$$
0 \eqdef d_0 
< d_1  < \cdots < d_{\ell-1}  \leq 
 k-1,
$$    
\item 
and  $\ell$
symmetric matrices   $P_i\in \RR^{n\times n}$
of respective inertia $(d_{i},0,n-d_{i})$,
with  $i\in \{0, \ldots, \ell-1\}$,
\end{itemize}
 such that
\begin{subequations}
\label{eqn:corollaryLMI}
\begin{align}\label{eqn:corollary++LMI}
    &A^\top P_i  + P_i A    \prec 2\mu_i P_i, \qquad \forall i\in \{0, \dots, \ell-1\},\\
    \label{eqn:corollary++mu}
    &\sum_{i=0}^{\ell-1}{h_{i}} \, \mu_i \leq 0\ ,
    \end{align}
\end{subequations}
where $h_0\geq 1$ and 
$h_i = d_{i+1} - d_{i}$, for all $i=\{0, \ldots,\ell-1\}$ with $d_\ell\in\NN$ satisfying 
$ d_{\ell - 1 }+1\leq d_\ell \leq  k$.
\end{theorem}

The proof of  Theorem~\ref{thm:k_contraction_LMI_cLTI} is postponed to Section~\ref{subsec:proof_theorem_cLTI}.
To provide some intuition relative to Theorem~\ref{thm:k_contraction_LMI_cLTI}, we anticipate that the constants $\mu_i$ are  bounding the real part of the eigenvalues of the matrix $A$. That is, $\mu_0$ bounds the eigenvalue with largest real part of $A$, while $\mu_1$ bounds the second largest eigenvalue with real part different from the first, and so on. Then, \eqref{eqn:corollary++mu} can be interpreted as a bound in the partial sum of eigenvalues of $A$, considering their multiplicities, in turn implying $k$-contraction (see the proof in Section~\ref{subsec:proof_theorem_cLTI} for more details).

\begin{remark}
A particular case in which the former Theorem applies is when $\ell=1$. In that case, the former condition reduces to the existence of a real number $\mu_0\leq 0$ and a symmetric positive definite matrix $P_0\succ0$ such that 
$
A^\top P_0 + P_0A \prec 2\mu_0 P_0\,. 
$
 This condition is satisfied if and only if $A$ is Hurwitz, which implies that system \eqref{eqn:linear_system} is  $k$-contractive for all $k\in\{1, \dots, n\}$ by means of Lemma~\ref{lem:k+1_contract}.
\end{remark}

\begin{remark}\label{rem:algorithm}
The inertia requirements in \eqref{eqn:corollary++LMI} cannot be represented as semidefinite constraints. 
However, these constraints can be dropped without a significant impact on the solution of the inequality. Indeed, by Lemma~\ref{lem:inertia_shifting} in Section~\ref{sec:proof_L},
a given constant $\mu_i$ imposes a specific inertia on the matrix $P_i$, depending on the eigenvalues of $A$. 
\end{remark}

\begin{remark}
    Other attempts at dropping the requirement of matrix compound computations appeared in the literature \cite{dalin2022verifying}. However, to the best of our knowledge, Theorem~\ref{thm:k_contraction_LMI_cLTI} is the first result proposing necessary and sufficient conditions \cite[Section IV.C]{zoboli2023lmi}.
\end{remark}

\subsection{Computational  burden  of Theorem~\ref{thm:k_contraction_LMI_cLTI}}\label{sec:computational_burden}
 We now compare Theorem~\ref{thm:demidovich} and Theorem~\ref{thm:k_contraction_LMI_cLTI} in terms of the computational burden imposed by the solution of the respective matrix inequalities. 
We focus on the result in Theorem~\ref{thm:k_contraction_LMI_cLTI} for the case $\ell=k$ and $d_i=d_{i-1}+1$, since it provides the largest set of matrix inequalities.
Let $M\in\RR^{r\times r}$ be an arbitrary square matrix and $Q\in\RR^{r\times r}$ be a symmetric matrix. Since $Q$ is symmetric, each condition of the form 
$
QM+ M^\top Q \preceq \mu Q
$
requires the computation of $N=r(r-1)/2+1$ variables, namely the entries of the top triangular portion of $Q$ and the scalar $\mu$. Then, Theorem~\ref{thm:demidovich} requires $N_1=\binom{n}{k}\big(\binom{n}{k}+1\big)/2+1$ variables, while Theorem~\ref{thm:k_contraction_LMI_cLTI} requires $N_2=kn(n-1)/2+k$ variables. 
\begin{figure}
    \centering
    \includegraphics[width=0.9\linewidth]{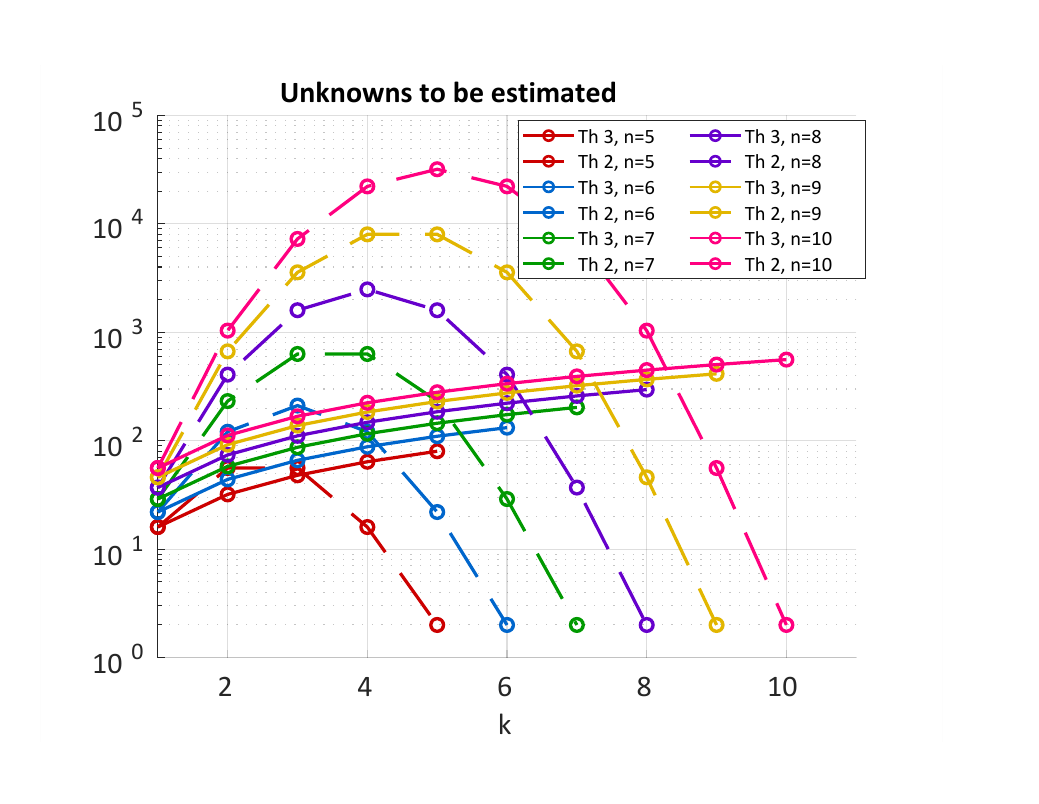}\vspace{-1em}
    \caption{Number of variables to be estimated by Theorem~\ref{thm:demidovich} (dashed) 
    and by  Theorem~\ref{thm:k_contraction_LMI_cLTI} (solid)  in function of $k$. Colors refer to different  $n$.}
    \label{fig:computational_compar}\vspace{-1em}
\end{figure}

To better understand how the size of the problem scales with different values of $k$ and $n$, 
 we refer to 
Fig.~\ref{fig:computational_compar}. Clearly, for large dimensional systems and low $k$, 
the conditions in Theorem~\ref{thm:k_contraction_LMI_cLTI}
 ask for a significantly smaller number of variables. 
Moreover, even in the worst case of $k=n$,  Theorem~\ref{thm:k_contraction_LMI_cLTI} typically requires between $10^2$ and $10^3$ variables.  Differently, Theorem~\ref{thm:demidovich} can easily reach $10^4$ variables. 

This computation shows that conditions in Theorem~\ref{thm:k_contraction_LMI_cLTI} do not grow in dimension as fast as the condition in Theorem~\ref{thm:demidovich}. Moreover,  for $k$ sufficiently smaller than $n$, we have $N_2\le N_1$.
We recall that the cases $k=\{1,2,3\}$ are interesting from a control viewpoint, see Section~\ref{sec:asymptotic}.

\subsection{$\mathit{k}$-order stabilizability}\label{sec:k_order_properties}
Consider a linear system of the form \eqref{eqn:linear_system_nonautonomous}. It is well-known that stabilizability of the pair $(A,B)$ is a necessary and sufficient condition for the existence of a stabilizing controller. A similar property can be defined when considering $k$-contractive designs. We refer to this property as \emph{$k$-order stabilizability}.

\begin{definition}[$\boldsymbol k$-order stabilizability]\label{def:k_stab}
    System \eqref{eqn:linear_system_nonautonomous} is $k$-order stabilizable if there exists a matrix $K\in\RR^{n\times m}$ such that the closed-loop system $\dot x = (A-BK)x$ is $k$-contractive.
\end{definition}

Conditions for $k$-order stabilizability can be easily derived by transforming the system into a suitable form.
Using standard Kalman decomposition,  system \eqref{eqn:linear_system_nonautonomous} is algebraically equivalent to a system of the form
\begin{equation}\label{eqn:controllable_decomposition}
    \begin{aligned}
      {\begin{bmatrix}
             \dot x_c \\  \dot  x_u
        \end{bmatrix}} &= \begin{bmatrix}
            A_{c} & A_{12}\\ 0 & A_{u}
        \end{bmatrix}\begin{bmatrix}
            x_c \\ x_u
        \end{bmatrix} + \begin{bmatrix}
            B_c \\ 0
        \end{bmatrix}u
    \end{aligned}
\end{equation}
where $x_c \in \RR^{n_c}$, $x_u\in \RR^{n_u}$, $n_c+n_u=n$ and the pair $(A_{c}, B_c)$ is controllable.  
The  non-negative integer  $n_u$ is the dimension of the null-space of the controllability matrix of \eqref{eqn:linear_system_nonautonomous}. Consequently, we admit the possibility of $n_u=0$ and $A_{u}$ being non-existing. 

\begin{lemma}\label{lem:k-order_stabilizability}
    System \eqref{eqn:linear_system_nonautonomous} is $k$-order stabilizable if and only if 
    its canonical Kalman decomposition 
    \eqref{eqn:controllable_decomposition}
    satisfies one  of the two following conditions:
    \begin{itemize}
        \item the integer $n_u$ satisfies $0\leq n_u < k$;
        \item the autonomous system $\dot x_u = A_{u}x_u,$ is $k$-contractive.    
    \end{itemize}
\end{lemma}
The proof of Lemma~\ref{lem:k-order_stabilizability} is postponed to Section~\ref{sec:proof_lem_stab}. Intuitively, Lemma 2 asks for the uncontrollable part to be already $k$-contractive (or of dimension smaller than $k$). For the case $k=1$, Lemma~\ref{lem:k-order_stabilizability} reduces to $n_u = 0$,  a necessary and sufficient condition for controllability in linear systems, or $\dot x_u = A_{u}x_u,$ being stable, which is a sufficient condition for the classical notion of stabilizability. 
We also remark that similar definitions could be developed for $k$-order controllability, observability, and detectability.

\subsection{$\mathit{k}$-contractive feedback design}\label{sec:linear_design}

Starting from a $k$-order stabilizability property and the decomposition \eqref{eqn:controllable_decomposition}, one can easily derive a $k$-contractive feedback design, e.g., via pole placement on the pair $(A_{c}, B_c)$. 
Nonetheless, in view of an extension of these notions
to the
nonlinear context, we look 
for coordinate-free conditions, i.e.  
that do not rely on 
change of coordinates and decompositions that would not be easy to extend to  nonlinear systems.

Motivated  by the result in Theorem~\ref{thm:k_contraction_LMI_cLTI}, we 
now  derive  constructive conditions for designing $k$-contractive controllers. This section presents a design methodology that follows the philosophy of feedback stabilization based on Lyapunov tests for stabilizability \cite[Section 14.5]{Hespanha09}. That is, we first solve a set of matrix inequalities, which are feasible if and only if the system is $k$-order stabilizable. Then, the controller is derived from the result of these inequalities. 
We start by presenting a  generalized Lyapunov test for $k$-order stabilizability. 
\begin{theorem}~\label{thm:k-order-stabilizability}
System \eqref{eqn:linear_system_nonautonomous} is $k$-order stabilizable
if and only if
there exist: 
\begin{itemize}
    \item a positive integer
$\ell \in \NN$ satisfying 
$1\leq \ell \leq k $, 
\item 
$\ell$
real numbers $\mu_i\in \RR$, with 
$i\in\{0, \ldots, \ell-1\}$, 
\item 
$\ell$ positive integers $d_i\in\NN$,
with 
$i\in \{0, \ldots, \ell-1\}$,
satisfying
$$
0 = d_0 
< d_1  < \cdots < d_{\ell-1}  \leq 
 k-1,
$$
\item 
and $\ell$
symmetric matrices   $W_i\in \RR^{n\times n}$
of respective inertia $(d_{i},0,n-d_{i})$,
with  $i\in \{0, \ldots, \ell-1\}$,
\end{itemize}
such that
\begin{subequations}\label{eqn:k_contraction_LMI_linear_stab2_total}
\begin{align}
    \label{eqn:k_contraction_LMI_linear_stab2+}
        &W_iA^\top  +  AW_i  -BB^\top   \prec 2\mu_i W_i,\quad \forall i\in \{0, \dots, \ell-1\}, \\
    &\sum_{i=0}^{\ell-1}{h_{i}} \, \mu_i \leq 0\ ,\label{eqn:mu_sum_condition_stab+}
    \end{align}
    \end{subequations}
where $h_0\geq 1$,
$h_i = d_{i+1} - d_{i}$, for all $i=\{1, \ldots,\ell-1\}$ with $d_\ell\in\NN$ satisfying 
$ d_{\ell - 1 }+1\leq d_\ell \leq  k$.
\end{theorem}
The proof of Theorem~\ref{thm:k-order-stabilizability} is postponed to Section~\ref{sec:k-order-stabilizability_proof}.
Notice that
for $k=1$, inequalities \eqref{eqn:k_contraction_LMI_linear_stab2+}- \eqref{eqn:mu_sum_condition_stab+} reduce to the existence of a constant $\mu_0\leq0$ and a symmetric positive definite matrix $W_0\succ 0$ such that
$$
\begin{aligned}
W_0A^\top  +  AW_0  -BB^\top    &\prec 2\mu_0W_0.
\end{aligned}
$$
Hence, we recover
the well-known Lyapunov test for stabilizability \cite[Section 14.4]{Hespanha09}. 
Differently put,
the inequalities  \eqref{eqn:k_contraction_LMI_linear_stab2_total} can be seen as a generalization of the Lyapunov test for stabilizability to the context of $k$-contraction.

Now, based on the presented generalized Lyapunov test for $k$-order stabilizability, we can directly derive a $k$-contractive feedback controller for the linear system \eqref{eqn:linear_system_nonautonomous}. The result is summarized in the following proposition.

\begin{proposition}~\label{pro:k-order-feedback}
Assume that \eqref{eqn:linear_system_nonautonomous} is $k$-order stabilizable.
Then, there exist: 
\begin{itemize}
    \item a positive integer
$\ell \in \NN$ satisfying 
$1\leq \ell \leq k$,
\item 
$\ell$
real numbers $\mu_i\in \RR$, with 
$i\in\{0, \ldots, \ell-1\}$, 
\item 
$\ell$ positive integers $d_i\in\NN$,
with 
$i\in \{0, \ldots, \ell-1\}$,
satisfying
$$
0 = d_0 
< d_1  < \cdots < d_{\ell-1}  \leq 
 k-1,
$$
\item 
and $\ell$
symmetric matrices   $W_i\in \RR^{n\times n}$
of respective inertia $(d_{i},0,n-d_{i})$,
with  $i\in \{0, \ldots, \ell-1\}$,
\end{itemize}
such that \eqref{eqn:k_contraction_LMI_linear_stab2_total} is satisfied, and  the following colinearity relation holds
\begin{equation}
    \label{eqn:colinear}
  B^\top W_i^{-1} = B^\top W_0^{-1}, \quad \forall \,i\in \{0, \ldots, \ell-1\}.
\end{equation}
Furthermore, with this solution,
system \eqref{eqn:linear_system_nonautonomous} is $k$-contractive with the feedback 
    law    
\begin{equation}\label{eqn:feedback_design+}
  u = - Kx, \qquad   K = \frac\rho2 B^\top W_0^{-1}, 
  \qquad \forall\, \rho\geq1.
    \end{equation}
\end{proposition}
The proof is postponed to Section~\ref{sec:proof_feedback}. 
We highlight that, since $k$-contraction for the closed-loop system is preserved for all $\rho\geq1$, the proposed controller presents a generalization of the infinite-gain margin property \cite[Section 3.2.2]{sepulchre2012constructive} to the framework of partial stabilization. Consequently, this result expands similar infinite-gain margin designs \cite{giaccagli2023lmi} from $1$-contraction to $k$-contraction.

\section{$k$-contraction  for nonlinear systems}\label{sec:nonlin}
As a follow-up to the linear scenario, we now move to the analysis of $k$-contraction for nonlinear systems. The main goal is to provide sufficient conditions inspired by the results in Section~\ref{sec:lin}. 
\subsection{Sufficient conditions}
Consider a  nonlinear system of the form \eqref{eqn:original_system}. The following theorem provides sufficient conditions for $k$-contraction.

\begin{theorem}\label{thm:k_contraction_LMI}
Let $\cS\subsetneq \RR^n$ be a compact forward invariant set.
Suppose there exist two symmetric matrices $P_0, P_{1}\in \RR^{n\times n}$ 
 of respective inertia  $(0,0,n)$ and $(k-1,0,n-k+1)$,
  and  $\mu_0, \mu_1\in \RR$ such that 
\begin{subequations}\label{eqn:2_contraction_LMI}
\begin{align}\label{eqn:P0_LMI}
        &\dfrac{\partial f}{\partial x}(x)^\top P_0  + P_0 \dfrac{\partial f}{\partial x}(x)     \prec 2\mu_0 P_0,
\\ \label{eqn:Pk_LMI}
        &\dfrac{\partial f}{\partial x}(x)^\top P_1  + P_1\dfrac{\partial f}{\partial x}(x)    
 \prec 2\mu_1 P_1,
 \\ \label{eqn:mu_sum_condition_NL}
 &\mu_1 + (k-1)\mu_0 <0,
\end{align}
\end{subequations}
for all $x\in\cS$.
Then, system \eqref{eqn:original_system} is infinitesimally $k$-contractive on $\cS$ (therefore, $k$-contractive on $\cS$). 
\end{theorem}

A detailed discussion of Theorem~\ref{thm:k_contraction_LMI} is postponed to  Section~\ref{subsec:proof_theorem_k_contraction}, along with the associated proof.  Intuitively, inequality \eqref{eqn:P0_LMI} bounds the expansion rate for the variational system \eqref{eqn:variational_system} by a factor $\mu_0$. Differently, the second inequality \eqref{eqn:Pk_LMI} bounds the contraction rate for the variational system on a subsapce by a factor $\mu_1$. Consequently, inequality \eqref{eqn:mu_sum_condition_NL} constraints the contraction rate to be faster than the expansion rate. 
This resembles the eigenvalue bounding approach of Section~\ref{sec:lin}. However, a simple eigenvalue interpretation is not applicable in the nonlinear framework. Hence, we directly bound the fastest diverging direction and the slowest converging one, asking for the latter to be sufficiently fast through \eqref{eqn:mu_sum_condition_NL}.

Notice that Theorem~\ref{thm:k_contraction_LMI} considers constant matrices $P_0,P_1$. 
In view of 
recent results on Riemannian contraction analysis, 
e.g. \cite{SIMPSONPORCO201474,andrieu2020characterizations, Sanfelice2012}, we 
claim 
that  
the use of constant metrics is restrictive and that
state-dependant metrics can help  
in widening the result to more general cases. A direct consequence of this 
observation is the fact 
that, contrarily to the linear case, the conditions 
of Theorem \ref{thm:k_contraction_LMI}
are in general not equivalent to those of
Theorem \ref{thm:demidovich}.
Nonetheless, we highlight that restricting ourselves to the case of constant matrices can help to derive some new asymptotic behavior for $k$-contractive systems,  
as discussed at the end of Section~\ref{sec:matrix_compounds}.  Finally, we remark that, under linear dynamics, Theorem~\ref{thm:k_contraction_LMI}
recovers Theorem~\ref{thm:k_contraction_LMI_cLTI} for the particular case $\ell=2$ and $d_1=k-1$.

\begin{remark}\label{rem:convex_relaxation}
   Theorem~\ref{thm:k_contraction_LMI} requires solving an infinite 
    set  
   of matrix inequalities. Nonetheless, there are multiple strategies that can be used to 
   reduce it to a feasible problem. 
   For instance, one could exploit convex relaxation, 
   as explained in \cite[Section VI]{Forni2019}.  Alternatively, for systems with a semi-linear structure  (namely, $f(x) = Ax + g(x)$), one can obtain a finite set of LMIs if the nonlinear term satisfies a monotonic or a sector-bounded condition \cite{zoboli2024quadratic}.
\end{remark}

\subsection{Relaxing conditions for the planar case}
Note that, differently from Theorem~\ref{thm:demidovich}, in Theorem~\ref{thm:k_contraction_LMI} we require the set $\cS$ to be compact. We conjecture this compactness assumption can also be dropped in Theorem~\ref{thm:k_contraction_LMI}. This conjecture is motivated by the following result for the planar case $n=k=2$.

\begin{lemma}\label{lem:planar_case}
Let $\cS\subseteq \RR^2$  and assume there exist  symmetric matrices $P_0, P_{1}\in \RR^{2\times 2}$ 
 of inertia  $ \cine(P_0)=(0,0,2)$, 
 $\cine(P_{1})=(1,0,1)$
  and  $\mu_0,\mu_{1}\in \RR$ such that, 
  for all $x\in\cS$, inequalities in
 \eqref{eqn:2_contraction_LMI} are satisfied.
Then, system \eqref{eqn:original_system} is $2$-contractive on $\cS$. 
\end{lemma}
The proof of Lemma~\ref{lem:planar_case}  is postponed to Section~\ref{sec:proof_planar}.
This result shows that the inequalities in \eqref{eqn:2_contraction_LMI} may be valid on the whole $\RR^n$. Consequently,
in future works, we aim at exploring if 
Theorem~\ref{thm:k_contraction_LMI} can be expanded to the whole $\RR^n$. Currently, the technical obstruction that prevents us to conclude the conjecture is the use of  Theorem~\ref{thm:Magic_splitting} and Lemma~\ref{lemma:transition_comp_contr} in Theorem~\ref{thm:k_contraction_LMI} proof, which require $\cS\subsetneq\RR^n$ in order to guarantee a bounded invariant subspace splitting. 

\subsection{$\mathit{k}$-contractive feedback design} 
\label{sec:nonlin_feed}

Following the lines of the linear results presented in Section~\ref{sec:lin}, we now elaborate on the conditions for $k$-contraction proposed in Theorem~\ref{thm:k_contraction_LMI}. We aim at devising $k$-contractive controllers for nonlinear systems.  Precisely, we consider nonlinear systems of the form 
\begin{equation}\label{eqn:nonlinear_system_nonautonomous}
\begin{aligned}
    \dot x &= f(x)+ Bu 
\end{aligned}
\end{equation}
where $u\in \RR^m$ and $f$ is sufficiently smooth. 
In the next proposition, we provide a result on $k$-contractive controller design.

\begin{proposition}\label{thm:2_contraction_feedback_design}
 Let $\cS\subsetneq \RR^n$ be a compact set, and assume there exist two symmetric matrices $W_0,W_1\in \RR^{n\times n},\,$ with $W_0\succ0$ and inertia $\cine(W_1)=(k-1,0,n-k+1)$ and a pair of real numbers $\mu_0,\mu_1\in\RR$, such that, for all $x\in\cS$,
\begin{subequations}\label{eqn:k_contraction_LMI_NLdesign}
\begin{align}
\label{eqn:k_contraction_LMI_NL_stab}
        &W_0\tfrac{\partial f}{\partial x}(x)^{\!\top}  +  \tfrac{\partial f}{\partial x}(x)W_0  -BB^\top    \prec 2\mu_0 W_0 
\\
  &  \notag W_1\left(\tfrac{\partial f}{\partial x}(x)-\tfrac{1}{2}BB^\top W_{0}^{-1}\right)^{\!\!\top}
  \!\!\!\!+ \left(\tfrac{\partial f}{\partial x}(x)-\tfrac{1}{2}BB^\top W_{0}^{-1}\right)W_1
\\&\qquad\qquad -BB^\top \prec 2\mu_1 W_1.
\label{eqn:k_contraction_LMI_NL_stab3}
\end{align}
    \end{subequations}
    Then, there exists a real number
    $\omega>0$ such that if
    \begin{equation}\label{eqn:mu_sum_condition_NL_stab}
        (k-1)\mu_0+\mu_1+ \omega<0,
    \end{equation}
     the feedback law $u= -Kx$
    with 
\begin{equation}\label{eqn:feedback_design_NL}
    K = \tfrac{1}{2} B^\top (W_0^{-1}+W_1^{-1}).
    \end{equation}
    makes the system \eqref{eqn:nonlinear_system_nonautonomous}
    $k$-contractive on $\cS$, if $\cS$ is forward invariant for the closed-loop system.
\end{proposition}

The proof of Proposition~\ref{thm:2_contraction_feedback_design} is postponed to Section~\ref{sec:2_contractive_design_NL}. 
 Proposition~\ref{thm:2_contraction_feedback_design} is an extension of the result for linear systems in Theorem~\ref{thm:k_contraction_LMI_cLTI}  to the 
nonlinear framework.  
 However, besides the nonlinearities, we highlight some main differences between the two results. First, since we require constant matrices $W_i$, the nonlinear result cannot be proven to be necessary in general. Second, even if such constant matrices do exist, there is no guarantee that they satisfy a colinearity condition \eqref{eqn:colinear} uniformly on $x$.  Hence, Proposition~\ref{thm:2_contraction_feedback_design} proposes an alternative design that trades the colinearity condition \eqref{eqn:colinear} for conservativeness in the sum of rates \eqref{eqn:mu_sum_condition_NL_stab}, i.e. the addition  of $\omega>0$. A similar approach can be used in the linear scenario to avoid the colinearity condition \eqref{eqn:colinear} at the price of the necessity result.

\subsection{Asymptotic behavior}\label{sec:asymptotic}

In this section, we describe some asymptotic behaviors 
for $k$-conctractive systems.
In particular, for the cases 
$k=\{1,2\}$
we can directly rely on Definition~\ref{def:k_contraction}
to establish the existence of a unique or multiple equilibria 
and exclude more complex behaviors (such as limit cycle or chaotic ones). For $k=3$, the sufficient conditions of Theorem~\ref{thm:k_contraction_LMI} allow to conclude the existence of simple attractors, that is, fixed points or limit cycles. Finally, for $k=n$, the sufficient condition of Theorem~\ref{thm:demidovich} allows to exclude the existence of repelling equilibria.

\smallskip
\subsubsection*{$\mathit1$-contractivity}
For $k=1$, we recover the asymptotic behavior of the classical notion of contraction  
(see, e.g. \cite{bullo2022contraction}
and references therein). Namely, one can deduce the existence and attractiveness of a unique equilibrium.

\begin{lemma}\label{lem:standard_contraction}
    Assume that system~\eqref{eqn:original_system} is $1$-contractive (according to 
    Definition~\ref{def:k_contraction}) in a closed  and forward invariant set $\cS\subseteq\RR^n$. Then, system~\eqref{eqn:original_system} has a unique equilibrium which is exponentially stable 
    with a domain of attraction including $\cS$.
\end{lemma}

Noting that for 
$k=1$ 
volumes boil down to distances, the previous lemma can be proved by 
using Banach fixed point theorem as in   \cite[Lemma 2]{giaccagli2022sufficient}. 
Furthermore, we highlight that 
$\cS$ is not required to be a compact set, 
hence it can be equal to $\RR^n$.
\smallskip
\subsubsection*{$\mathit2$-contractivity}
The property of $2$-contraction can be used to study  phenomena: 
non-existence of trivial periodic solutions and multi-stability.

\begin{lemma}\label{lem:non_periodic}
    Assume that system~\eqref{eqn:original_system} is $2$-contractive (according to 
    Definition~\ref{def:k_contraction})
    in a compact and forward invariant set $\cS\subsetneq\RR^n$. Then, the following holds:
    \begin{itemize}
        \item[1)] \cite{Muldowney1990}
system~\eqref{eqn:original_system} has no non-trivial periodic solutions in $\cS$;
        \item[2)] \cite{Michael1995}
        every solution of system~\eqref{eqn:original_system} initialized in $\cS$ converges to an equilibrium point (which may not be the same for every solution).
    \end{itemize}
\end{lemma}

The first item of the previous lemma shows that  $2$-contraction and the conditions in Theorem~\ref{thm:k_contraction_LMI} could be understood as a generalization of the Bendixson's criteria for systems of dimension larger than 2. 
The 
second item  connects $2$-contraction with multi-stability.

\smallskip
\subsubsection*{$\mathit{3}$--contractivity}

For the case $k=3$, the property of $3$-contraction can be used to study non-trivial periodic solutions of the system. In this case, however, we don't explicitly rely on the Definition~\ref{def:k_contraction}, but rather on the (more stringent) sufficient conditions of 
Theorem~\ref{thm:k_contraction_LMI}.

\begin{lemma}\label{lem:3-contraction}
    Assume that  system~\eqref{eqn:original_system} satisfies the conditions of Theorem~\ref{thm:k_contraction_LMI}  
    with $k=3$
        in a compact forward invariant set $\cS\subsetneq\RR^n$. Then, every solution of \eqref{eqn:original_system} with initial condition in $\cS$ converges to a simple attractor, namely, a fixed point or a limit cycle (which may not be the same for every solution).
\end{lemma}

The proof of Lemma~\ref{lem:3-contraction} is postponed to Section~\ref{sec:proof_3_contraction}.
We remark that this result is not generally true for $3$-contractive systems  and it is a consequence of the condition \eqref{eqn:2_contraction_LMI} with constant matrices $P_i$. As an example, 
consider the Rössler system \cite{Rossler1979}
\begin{equation}\label{eqn:Rossler}
\begin{aligned}
 &\dot x_1  =  x_2,\quad 
 \dot x_2  = -x_1-x_3, \\ 
 &\dot x_3  = 0.5((x_1-x_1^2)-x_3).
\end{aligned}
\end{equation}
 The 3-additive compound of the Jacobian is $\frac{\partial f}{\partial x}(x)^{[3]}=-0.5$. Therefore, condition \eqref{eqn:compound_LMI} is trivially satisfied and the  system \eqref{eqn:Rossler} is $3$-contractive. Nonetheless, the nonlinear term $x_1^2$  prevents the existence of a constant matrix $P_1$ in \eqref{eqn:Pk_LMI}. Therefore, even if the system is $3$-contractive and  evolves in a compact set, \eqref{eqn:2_contraction_LMI} cannot be satisfied and there is no guarantee it will converge to a fixed point or limit cycle. 
 Indeed, this system presents chaotic behavior and its trajectories do not converge to a simple attractor, see Figure~\ref{fig:rossler}.

 Alternatively, consider the following system,
\begin{equation}\label{eqn:Rossler_mod}
\begin{aligned}
 &\dot x_1  =  x_2-2x_3, \quad
 \dot x_2  = -x_1-x_3, \\ 
 &\dot x_3  = 0.5((x_1-x_1^3)-x_3).
\end{aligned}
\end{equation}
In this case, the 3-additive compound of the Jacobian 
 is $\frac{\partial f}{\partial x}(x)^{[3]}=-0.5$,  similarly to \eqref{eqn:Rossler}. Then,  
the system is 
$3$-contractive.  Moreover, it evolves in a compact forward invariant set.  However, differently from \eqref{eqn:Rossler}, 
inequalities \eqref{eqn:P0_LMI}-\eqref{eqn:mu_sum_condition_NL} are satisfied with $\mu_0 = 0.2, \mu_1 = -0.45$ and $P_0 = \smallmat{
    1.45 &    0.20  &  0.13 \\
    0.20  &  2.09  &  0.26 \\
    0.13  &    0.26 &    0.67},\,
 P_1 = \smallmat{
    -2.05 &   -1.02 &   -0.45 \\
   -1.02 &   26.75 &   14.47 \\
   -0.45 &   14.47 &    6.41}.$
Therefore, by Lemma~\ref{lem:3-contraction} we can conclude that the system will converge to a simple attractor (which may not be the same for every solution). 
 This behavior is shown in  
Figure~\ref{fig:modified_rossler}, which presents the evolution of three trajectories.

\begin{figure}[tb]
	\begin{center}
		\includegraphics[width=1\linewidth]{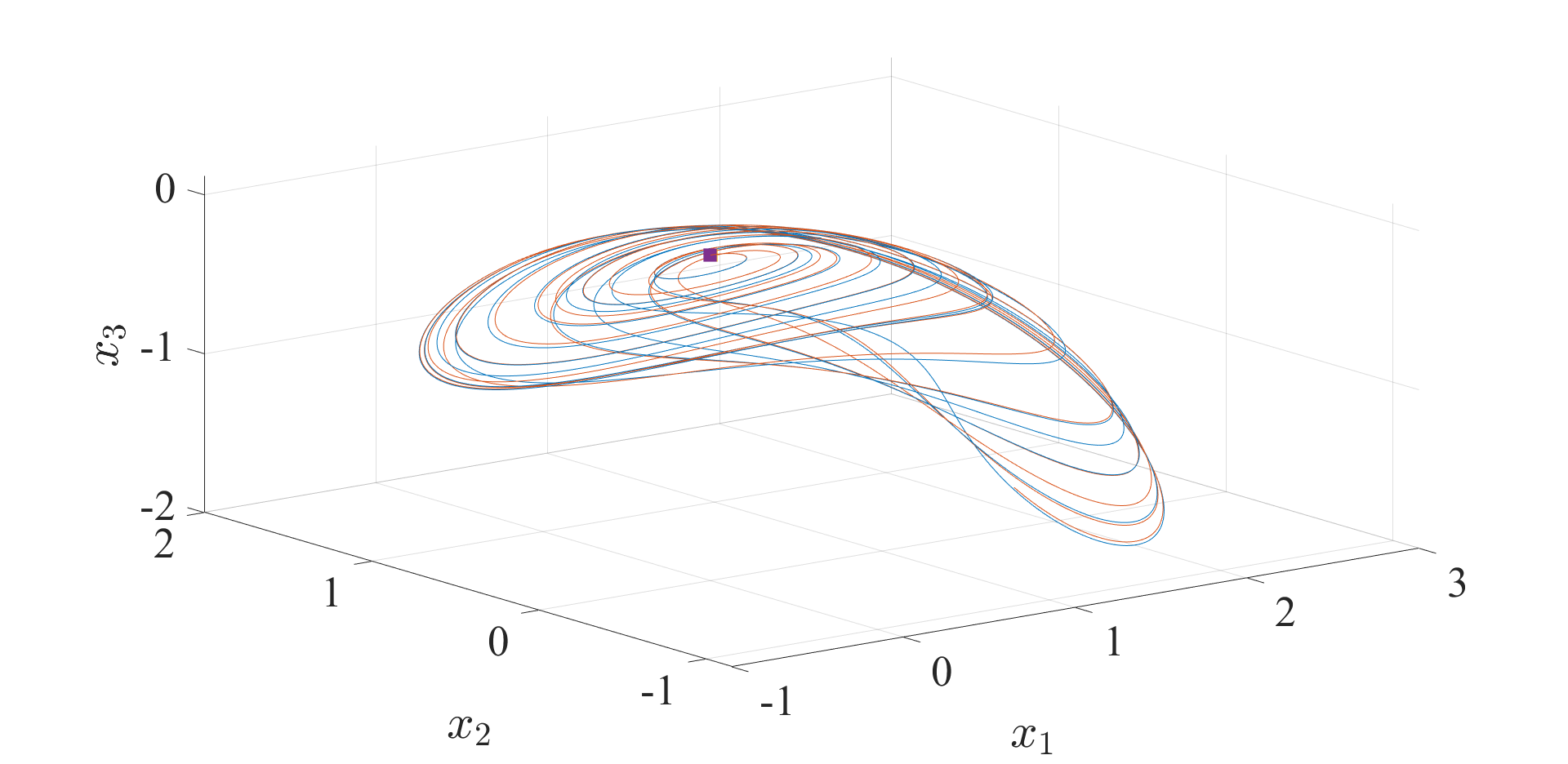}
		\caption{Evolution of two trajectories of the system \eqref{eqn:Rossler}. The first (blue) has an initial condition $[0.1,0.1,0]$ and the second (red) $[0.099,0.1,0]$.
  } \label{fig:rossler}
	\end{center}
\end{figure}

\begin{figure}[tb]
	\begin{center}
		\includegraphics[width=1\linewidth]{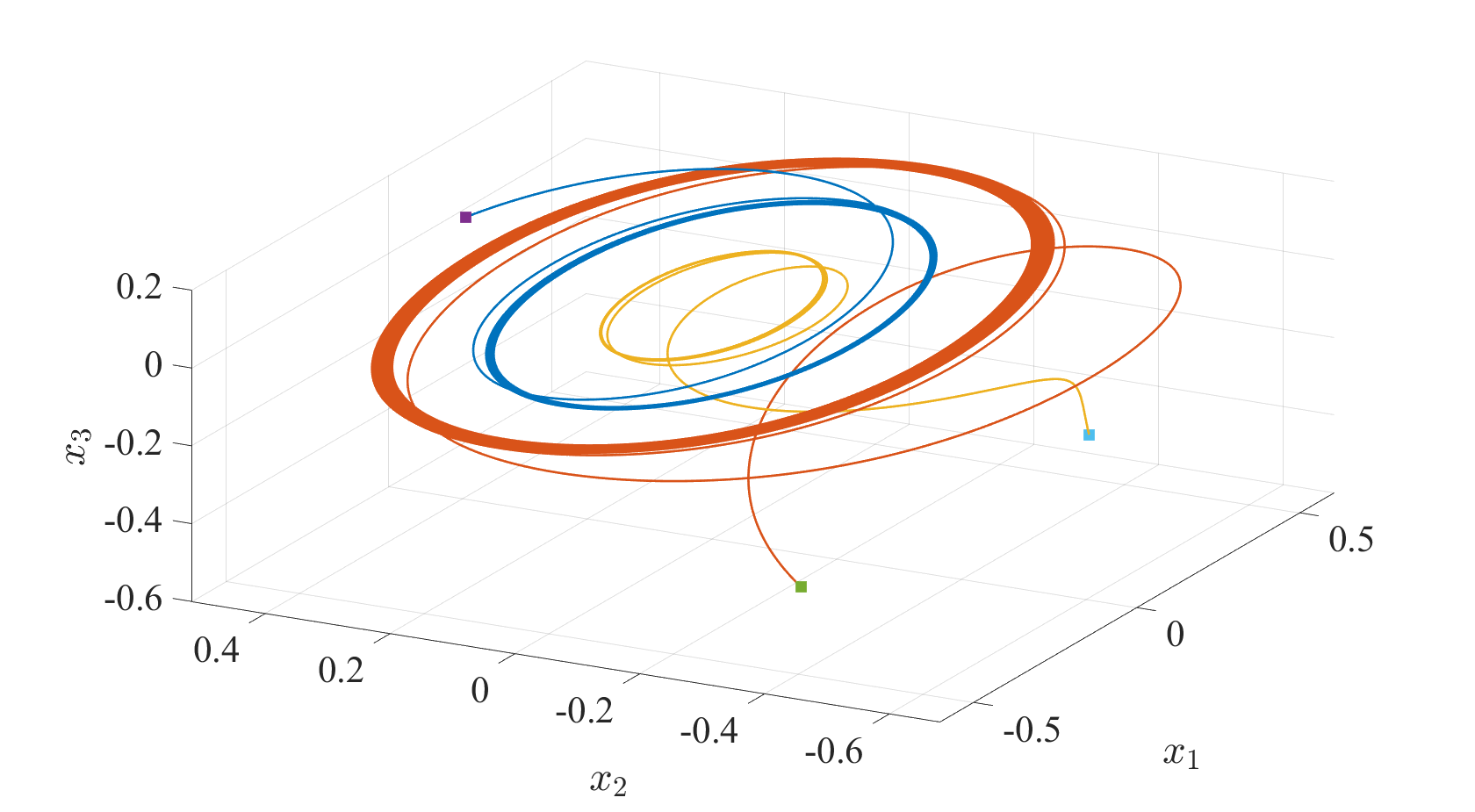}
		\caption{Evolution of three trajectories of the system \eqref{eqn:Rossler_mod}. The first (blue) has an initial condition $[0.2,0.5,0]$, the second (red) $[-0.3,-0.3,-0.5]$ and  the third (yellow)  $[0.2,-0.5,-0.3]$.} \label{fig:modified_rossler}
	\end{center}
\end{figure}

\begin{remark}\label{rem:reviewer_com}
The sufficient conditions in Theorem~\ref{thm:k_contraction_LMI}  resemble conditions for dominance analysis through cone fields, see, e.g., \cite{SMITH1987265, Forni2019}. However, this does not imply the topics are strongly related. For instance,  the proofs of
the asymptotic properties of $2$-contractive systems and systems that admit an invariant cone of
rank $1$ are very different. In the first case, one first proves that there are no periodic solutions, see Lemma~\ref{lem:non_periodic}, and then the Pugh's closing-lemma implies that any bounded solution converges to
an equilibrium. In the second case, (almost all) the trajectories of the system can be projected in
one-to-one way onto a $1$-dimensional
linear space, which implies that (almost all) trajectories converge to an equilibrium.
\end{remark}

\smallskip
\subsubsection*{$\mathit{n}$-contractivity}
Finally, for the case $k=n$, one can exclude the existence of repelling equilibria, that is, equilibria such that any infinitesimally small perturbation acting in an arbitrary direction results in a diverging trajectory. This is formalized as follows.

\begin{definition}[Infinitesimally repelling equilibrium]
    \label{def:repelling}
    Given a dynamical system $\dot x = f(x)$, we say that $x^\circ\in\RR^n$ is an 
    infinitesimally 
    repelling equilibrium if $f(x^\circ)=0$ and 
     $-\tfrac{\partial f}{\partial x}(x^\circ)$ is Hurwitz.
\end{definition}

\begin{lemma}\label{lem:n-contraction}
    Assume that system~\eqref{eqn:original_system} satisfies the conditions of Theorem~\ref{thm:demidovich} with $k=n$ in a set $\cS\subseteq\RR^n$. Then, system \eqref{eqn:original_system} has no 
    infinitesimally 
    repelling equilibria in $\cS$.
\end{lemma}
The proof of Lemma~\ref{lem:n-contraction} is postponed to Section~\ref{sec:proof_n_contraction}. A similar result could be obtained using the conditions of Theorem~\ref{thm:k_contraction_LMI}. However, 
for $k=n$, the conditions of Theorem~\ref{thm:k_contraction_LMI} are more restrictive and computationally heavier, see Section~\ref{sec:computational_burden}. 
Therefore, we omit the corresponding statement.

\section{Illustrations}\label{sec:Illustrations}

\subsection{Multi-stability of a grid-connected synchronverter}

Synchronverters are electrical inverters controlled 
to behave as a synchronous generator. In this section, we will consider the fourth-order model of the synchronverter presented in \cite{Lorenzetti2022},
\begin{equation}\label{eqn:synchronverter}
    \begin{aligned}
        \dot x_1 &= -\dfrac{R}{L}x_1 + x_2x_3 + \dfrac{V}{L}\sin{x_4} \\
        \dot x_2 &= -x_1x_3 -\dfrac{R}{L}x_2 -\dfrac{m}{L}i_fx_3 + \dfrac{V}{L}\cos{x_4} \\
        \dot x_3 &= \dfrac{m}{J}i_fx_2 - \dfrac{D_p}{J}(x_3-w_n) + \dfrac{T_m}{J} \\
        \dot x_4 &= x_3-w_n. 
    \end{aligned}
\end{equation}
Due to space restrictions, we
omit the physical meaning of the equations, and we refer to \cite{Lorenzetti2022} for a  detailed presentation.

Synchronverters properly operate when synchronized to the grid. 
Using the notation in \eqref{eqn:synchronverter}, this occurs when $x_3 = w_n$. 
With this in mind, one may wonder if for any initial condition (inside a predefined set) the synchronverter will converge to a state of synchronization. Note that the $x_4$ dynamics are just an integral of the error $x_3-w_n$. Hence, any equilibrium point of \eqref{eqn:synchronverter} presents the required synchronization behavior.

We highlight that \eqref{eqn:synchronverter} may present multiple equilibrium points \cite{Lorenzetti2022}. Therefore, the study of the eventual system's synchronization to the grid boils down to the study of its convergence properties to any of its equilibria.
Nonetheless, the complexity of such analysis led to several works dedicated to answering this specific question \cite{natarajan2018almost}.

In this section, we will use $k$-contraction as a simpler tool to prove the convergence properties of the synchronverter model \eqref{eqn:synchronverter}.  In particular, we will
show that Theorem~\ref{thm:k_contraction_LMI} can be used to prove $2$-contractivity of system \eqref{eqn:synchronverter}. Then, by boundedness of the system's trajectories, the convergence of the system to a (possibly non-unique) equilibrium point is guaranteed by item 2) of  Lemma~\ref{lem:non_periodic}. 

We consider the model parameters in \cite{Lorenzetti2022}, that is  $w_n = w_g = 100\pi, \ V= 230\sqrt{3}, \ J = 0.2, \ R = 1.875, \ L = 0.05675, \ Dp = 10, \ m = 3.5$, with $T_m=0,i_f = 1$. Moreover, we consider the forward invariant set $\cS$ defined as $x_1 \in\begin{bmatrix}-81 & 5\end{bmatrix}, \ x_2 \in\begin{bmatrix}-67 & 10.5\end{bmatrix}, \ x_3 \in\begin{bmatrix}298 & 315\end{bmatrix}, \ x_4 \in\begin{bmatrix}1 & -0.2\end{bmatrix}$. Then, by means of the convex relaxation approach mentioned in Remark~\ref{rem:convex_relaxation},
the matrix inequalities \eqref{eqn:P0_LMI}-\eqref{eqn:mu_sum_condition_NL} can be proven to hold with with $\mu_0 = 14, \mu_1 = -15$ and
$P_0 =  \smallmat{
        0.00007 &  0.000001 & -0.0005 &  0.00145 \\
   0.000001 &  0.00006 & -0.000076 &  0.001036 \\
  -0.0005 & -0.000076 &  0.011 & -0.02135 \\
   0.00145 &  0.001036 & -0.02135 &  0.9882 \\
    }, \,
    P_1 =  \smallmat{
        0.0007 &  0.00001 & -0.00851 &  -0.0692 \\
   0.00001 &  0.00027 & -0.000382&  -0.00009 \\
  -0.00851 & -0.000382 &  0.157&  1.308 \\
  -0.0692 & -0.00009 &  1.308 &  0.842 \\
    }.
$ Therefore, by Theorem~\ref{eqn:2_contraction_LMI}, \eqref{eqn:synchronverter} is $2$-contractive on $\cS$  and the desired stability property follows by item 2) of Lemma~\ref{lem:non_periodic}.
\color{black}

\subsection{Multi-stability via $2$-contractive feedback}
Consider the following 
nonlinear system:
\begin{equation}\label{eqn:example_controller}
\begin{aligned}
\dot x_1 = x_2 -x_3, \,
\dot x_2 =-x_1-x_3 + u, \,
\dot x_3 = x_1(x_1^2-0.25).
\end{aligned}
\end{equation}
In the absence of input ($u=0$), the system trajectories present non-trivial periodic solutions and the 3-additive compound is  $\frac{\partial f}{\partial x}(x)^{[3]}=0$.  Therefore, the system is not 3-contractive (nor 2-contractive nor 1-contractive).
The objective is to design a linear state-feedback controller $u=-Kx$, such that the system converges to a non-necessarily unique equilibrium point.

The closed-loop system has 3 equilibrium points, 
\begin{align*}
    &x_1^* = \{0,-0.5,0.5\}, \quad 
    x_2^* = x_3, \quad x_3^* = \dfrac{-(1+k_1)}{1+k_3+k_2}x_1,
    \end{align*}
where $k_1,k_2,k_3$ are the components of the feedback gain $K$. The existence of multiple equilibrium points prevents the design of a $1$-contractive linear controller, e.g., \cite{giaccagli2023lmi}. 
Nonetheless, it may still be possible to design a controller that  guarantees   
$2$-contraction 
 of the closed-loop. 
Since the system dynamics evolve a in compact set, 2-contraction guarantees convergence to a (possibly non-unique) equilibrium point by means of
item 2) of Lemma~\ref{lem:non_periodic}.

 With the aim of exploiting the results in Proposition~\ref{thm:2_contraction_feedback_design}, via the convex relaxation technique mentioned in Remark~\ref{rem:convex_relaxation} one can verify that the matrix inequalities \eqref{eqn:k_contraction_LMI_NL_stab}-\eqref{eqn:k_contraction_LMI_NL_stab3} are satisfied with $\mu_0 = 0.3, \mu_1 = -0.6$ and
$W_0 =  \smallmat{
        1.86  & -1.21  & -0.92\\
   -1.21 &   2.13  &  0.86\\
   -0.92  &  0.86  &  0.96
    }, \,
    W_1 =  \smallmat{
        2.84  &  0.06  &  2.65\\
    0.06 &   0.27  & -5.16\\
    2.65  & -5.19 &  -2.29
    }.
$
With this specific selection, we have $\omega = 0.048$, which implies \eqref{eqn:mu_sum_condition_NL_stab} is satisfied. 
Therefore, according to Proposition~\ref{thm:2_contraction_feedback_design} the state-feedback law $u=-Kx$ with 
 gain  \eqref{eqn:feedback_design_NL}, namely $K = \smallmat{0.89   & 2.16  & -1.18}$,
makes the closed-loop system $2$-contractive. Indeed, the closed-loop system satisfies \eqref{eqn:compound_LMI} with $\eta =0.091$ and
$
Q =  \smallmat{
    1.26  & -0.06  &  0.46\\
   -0.05  &  2.74  & -1.34\\
    0.41  & -1.34  &  2.33
},
$
which validates the result by means of Theorem~\ref{thm:demidovich}. More precisely, the closed-loop system presents 3 equilibrium points, one unstable at the origin and $2$ (locally) asymptotically stable.

\color{black}

\section{Proof of linear results (analysis)}\label{sec:proof_L}

In the rest of this section, given a matrix $A\in \RR^{n\times n}$, we order its eigenvalues in such a way that
\begin{equation}\label{NEQ:numbering}
\Re(\lambda_1)\geq\Re(\lambda_2)\geq\dots\geq\Re(\lambda_n).
\end{equation}

\subsection{$\mathit{k}$-contraction properties and inertia theorem}

The following result is a direct consequence of the additive compound definition.

\begin{lemma}\label{lem:spectral-k-properties}
A system $\dot x =Ax$, with $x\in \RR^{n}$
is $k$-contractive if and only if the eigenvalues $\lambda_i$ of the matrix $A$ satisfy
\begin{equation}\label{eqn:necessary_cond_example}
    \sum_{i =1}^{k}\Re{(\lambda_{i})}<0 \ ,
\end{equation}
according to the above mentioned ordering \eqref{NEQ:numbering}.
\end{lemma}

\begin{proof}
 A necessary and sufficient condition for $k$-contraction of a linear system $\dot x =Ax$ is that the matrix $A^{[k]}$ is Hurwitz \cite{Wu2022}.  The necessity part can be proven by adapting \cite[Lemma 1]{zoboli2024k} to continuous time. Moreover, recall that a spectral property of the additive compound matrix is that the eigenvalues of the matrix $A^{[k]}$ are all the possible sums of the form $\lambda_{i_1}+\lambda_{i_2}+\dots+\lambda_{i_k}$, with $ 1 
\leq i_1<\dots < i_k \leq n$, see \cite{Wu2022}. That is, a necessary and sufficient condition for $k$-contraction is that
the sum 
of any combination of $k$ eigenvalues 
of $A$ is negative.
In particular, this holds true if and only if the 
first $k$ eigenvalues satisfy condition
\eqref{eqn:necessary_cond_example}
with  the ordering \eqref{NEQ:numbering}.
\end{proof}

Additionally, we state a lemma about inertia properties of the Lyapunov equation \eqref{eqn:corollary++LMI}, collecting
together various results from the literature, 
e.g., 
\cite[Lemma 1, Section 3]{smith1979poincare}, \cite[Theorem 2.5]{stykel2002stability}.

\begin{lemma}\label{lem:inertia_shifting}\itshape
Given a matrix $A\in \RR^{n\times n}$,  a real constant $\mu$, 
and an integer $p\in \{0,\ldots, n\}$,
the following statements are equivalent:
\begin{enumerate}
    \item $A$ has $p$ eigenvalues with real part larger than $\mu$
and $n-p$ eigenvalues with real part smaller than $\mu$,
    \item the matrix $A-\mu I$ has inertia $(n-p,0,p)$,
\item there exists a  symmetric matrix $P\in\RR^{n\times n}$ with 
inertia $\cine(P)=(p, 0 , n-p)$ satisfying
        $A^\top P  + P A   \prec 2\mu P$. 
    \item given a symmetric positive definite matrix 
$Q\succ 0$
   there exists a  symmetric matrix $P\in\RR^{n\times n}$ with inertia
 $\cine(P) =(p,0,n-p)$
    satisfying 
    $
 (A-\mu I)^\top P + P\ (A-\mu I) =  -Q.
$
\end{enumerate}
\end{lemma}

\begin{proof}
    The eigenvalues of $A-\mu I$ are the shifted eigenvalues 
    $\lambda_1-\mu, \dots,  \lambda_n-\mu$, 
    ordered as in \eqref{NEQ:numbering}.
    If $A-\mu I$ has inertia $(n-p,0,p)$,
    it implies that $\Re(\lambda_{p+1})-\mu<0$.
    This shows that  $(1)$ and $(2)$ are equivalent.
    The equivalence between $(3)$ and $(2)$
    is due to \cite[Lemma 1, Section 3]{smith1979poincare} and
    the one between $(2)$ and $(4)$
    is due to \cite[Theorem 2.5]{stykel2002stability}.
    Finally, we have that  $(4)$ implies  $(3)$.
    Indeed, since $Q\succ 0$, $(4)$ implies
    $
    A^\top P + P A = -2\mu P -Q \prec -2\mu P.
    $
    \end{proof}

\subsection{Proof of Theorem~\ref{thm:k_contraction_LMI_cLTI}}\label{subsec:proof_theorem_cLTI}

The proof of 
Theorem~\ref{thm:k_contraction_LMI_cLTI}
is obtained by combining the two previous lemmas.
To begin with, we introduce a new notation in order to represent  the eigenvalues of $A$ and their associated multiplicities. Precisely, consider the matrix $A$ in  \eqref{eqn:linear_system} and
let $\Pi:\CC\to\RR$ denote the canonical projection onto the real axis. Let $\sigma(A)$ be the spectrum of $A$ and suppose
\(
\Pi(\sigma(A))
= \{\alpha_1, \alpha_2, \dots, \alpha_q\}
\)
($q\le n$)
with 
\(
\alpha_1 > \alpha_2 > \dots > \alpha_q.
\)
Set
\(
\bar h_i = \card \big(\Pi^{-1}(\alpha_{i+1}) \bigcap \sigma(A)\big),
\)
where eigenvalues have been counted with their algebraic multiplicities
(so that $\bar h_0+ \bar h_1 + \cdots + \bar  h_{q-1} = n$). Finally, let $\bar d_0=0$ and
$
\bar d_i =\sum_{j=0}^{i-1}\bar h_j$,
for each $i\in\{1, \dots, q-1\}
$. 

We remark that the variables $\bar d_i,\bar h_i$ are related to variables $d_i,h_i$ of Theorem~\ref{thm:k_contraction_LMI_cLTI}, hence the similar notation. To see this relation, notice that  $\bar h_i$ represents the number of eigenvalues that project to $\alpha_{i+1}$ and $\bar d_i$ represents the amount of eigenvalues with real part strictly larger than $\alpha_{i+1}$. Therefore,  for any constant $\mu\in\RR$ such that $\alpha_1<\mu$, we have $\cine(-A+\mu I)=(0,0,n)$. Similarly, if $\alpha_{i+1}<\mu<\alpha_{i}$, we have that $\cine(-A+\mu I)=(\bar d_i,0,n-\bar d_i)$. Consequently, if we avoid the singular case $\pi_0(-A+\mu I)>0$, the matrix $-A+\mu I$ can only present a particular set of inertia $(\bar d_i,0,n-\bar d_i)$ with $i\in\{0,\dots,q-1\}$. From this fact and Lemma~\ref{lem:inertia_shifting}, we obtain that the matrix inequalities in \eqref{eqn:corollary++LMI} are only feasible for the particular set of inertia $\cine(P_i) = (\bar d_i,0,n-\bar d_i)$ with $i\in\{0,\dots,q-1\}$. A direct consequence of this result is that $\ell \leq q$.

We now present the main arguments proving sufficiency and necessity of the result in Theorem~\ref{thm:k_contraction_LMI_cLTI}.

 \noindent
\emph{Sufficiency.}
In order to prove the sufficiency, 
we will show that the set of inequalities
\eqref{eqn:corollaryLMI}
implies 
the condition 
\eqref{eqn:necessary_cond_example}.
To this end, notice that a solution of \eqref{eqn:corollary++LMI} for some $\mu_i\in\RR$ and $P_i$ with inertia $\cine(P_i) = (d_i,0,n-d_i)$ implies that $\cine(-A+\mu I) = (d_i,0,n-d_i)$ by means of Lemma~\ref{lem:inertia_shifting}. That is, $A$ has only $d_i$ eigenvalues with real part strictly larger than $\mu_i$. This is equivalent to the bound 
\begin{equation}\label{eqn:eig_bound}
\Re(\lambda_{d_{i+1}})\leq\Re(\lambda_{d_{i}+1})<\mu_i, \quad \forall i\in \{0, \dots, \ell-1\}.
\end{equation}
Now, due to Lemma~\ref{lem:inertia_shifting}, we have that $\mu_{i+1}<\mu_i$ for all $i\in\{0,\dots,\ell-2\}$. Therefore, the bound \eqref{eqn:corollary++mu} implies $\mu_{\ell-1}\leq0$, which combined with \eqref{eqn:eig_bound}, implies $\lambda_{d_\ell}<0$. Additionally,
 because the eigenvalues are ordered as per \eqref{NEQ:numbering},
the following bound trivially holds for all  $i\in \{0, \dots, \ell-1\}$
$$
\sum_{j =d_{i}+1}^{d_{i+1}}\Re{(\lambda_{j})}\leq (d_{i+1}-d_i) \Re{(\lambda_{d_{i}+1})}= h_i \Re{(\lambda_{d_{i}+1})},
$$
where the definition $ h_i  = (d_{i+1}-d_i)$ has been used.
Combining this bound with the fact that $d_\ell \leq k$, the bound $\lambda_{d_\ell}<0$ and the fact that the eigenvalues are ordered as in \eqref{NEQ:numbering}, we obtain the following,
\begin{align}
\sum_{i =1}^{k}\Re{(\lambda_{i})}
\leq \sum_{i =1}^{d_\ell}\Re{(\lambda_{i})}
&=\sum_{i =0}^{\ell-1}\sum_{j =d_{i}+1}^{d_{i+1}}\Re{(\lambda_{j})}
\nonumber \\
&\leq \sum_{i =0}^{\ell-1} h_i \Re{(\lambda_{d_{i}+1})}.
\label{eqn:eqn:eig_sum_bound}
\end{align}
Then, combining \eqref{eqn:eig_bound},  \eqref{eqn:eqn:eig_sum_bound} and \eqref{eqn:corollary++mu}  we have   
\begin{equation}
    \sum_{i =1}^{k}\Re{(\lambda_{i})} < \sum_{i=0}^{\ell-1}{h_{i}} \, \mu_i\le 0\, .
\end{equation}     
Consequently, the system is $k$-contractive  by Lemma~\ref{lem:spectral-k-properties}.

\noindent
\emph{Necessity.}
Define
\begin{equation}\label{eqn:allowed_pdom_indices}
    \begin{aligned}
    p_k
    &\defeq \max \left(\{\bar d_0, \bar d_1, \dots, \bar d_{q-1}\}\bigcap [0, k-1]\right), \\
    c_k
    &\defeq \card \left(\{\bar d_0, \bar d_1, \dots, \bar d_{q-1}\}\bigcap [0, k-1]\right).
    \end{aligned}
\end{equation}
Then, the following equality holds
\begin{equation}\label{eqn:alpha_equality}
(k-p_k)\alpha_{c_k} +\sum_{i=0}^{c_k-2}{h_{i}} \, \alpha_{i+1} =\sum_{i =1}^{k}\Re{(\lambda_{i})}.
\end{equation}
Hence, combining Lemma~\ref{lem:spectral-k-properties} and \eqref{eqn:alpha_equality}, if the system is $k$-contractive, the next bound is satisfied
\begin{equation}\label{eqn:necessary_cond_alpha}
(k-p_k)\alpha_{c_k} +\sum_{i=0}^{c_k-2}{h_{i}} \, \alpha_{i+1} < 0.
\end{equation}
Then, by continuity, there exists a scalar $\varepsilon>0$, such that
$$
    (k-p_k) (\alpha_{c_k}+\varepsilon) +\sum_{i=0}^{c_k-2}{h_{i}} \, (\alpha_{i+1}+\varepsilon) \le 0\,.
$$ 
Next, fix $\ell = c_k$, $d_\ell = k-p_k+d_{\ell-1}$, $d_i = \bar d_i$ for all $i\in\{0,\dots, \ell-1\}$ and select
$
\mu_{i-1}=\varepsilon + \alpha_{i},$ for all $ i\in\{1,\dots,\ell\}
$. We have 
\begin{align*}
\sum_{i=0}^{\ell-1}{h_{i}}\mu_i&=(k-p_k) (\alpha_{c_k}+\varepsilon) +\sum_{i=0}^{c_k-2}{h_{i}} \, (\alpha_{i+1}+\varepsilon) \le 0,
\end{align*}
thus showing \eqref{eqn:corollary++mu}.
Now, define matrices $\hat A_i \defeq A-\mu_i I$ with $i\in\{0,\dots,\ell-1\}$. It is clear that, since $\mu_{i-1}>\alpha_{i}$ by definition, each matrix $\hat A_i$ has  $\bar d_i$  eigenvalues with positive real part and $n-\bar d_i$  eigenvalues with negative real part. That is $\cine(\hat A_i) = (n-\bar d_i,0,\bar d_i)$. Then, by Lemma~\ref{lem:inertia_shifting}, there exist symmetric matrices $P_i$ with $\cine(P_i)=\cine(-\hat A_i) =(\bar d_i,0,n-\bar d_i)$ such that
$$
 A^\top P_i +  P_iA  \prec 2\mu_i P_i\qquad \forall i=0,\dots, \ell-1\,,
$$
 thus concluding the proof.

\section{Proof of linear results (design)}

This section follows the next notation. First, we will consider that every pair $(A,B)$, is algebraically equivalent to the form in \eqref{eqn:controllable_decomposition}. Second,
we define 
$
\Re(\lambda_1^u)\geq \dots \geq \Re(\lambda_{n_u}^u),
$
as the ordered set of eigenvalues of $A_{u}$.

\subsection{Proof of Lemma~\ref{lem:k-order_stabilizability}}\label{sec:proof_lem_stab}

 \noindent
\emph{Sufficiency.}
Without loss of generality, consider a system in the form \eqref{eqn:controllable_decomposition}.
After a feedback design $u= -Kx = -[K_c \ K_u](x_c^\top,x_u^\top)^\top$
is selected, the closed-loop eigenvalues
are given by $\sigma(A) = \sigma(A_c-B_c K_c)\cup \sigma(A_u)$.
Then,  we can arbitrarily assign the eigenvalues of the closed-loop matrix $A_c-B_c K_c$, since the pair $(A_c,B_c)$ is controllable. 
As a consequence, let $c =  k\lambda_1^u$ 
and select $K_c$ such that 
the largest eigenvalue of $A_c-B_c K_c$
has real part smaller than $-|\Re(c)|$.
Then, the conditions of Lemma~\ref{lem:spectral-k-properties} are satisfied 
either if $A_u$ is $k$-contractive or $n_u<k$ by construction.

 \noindent
\emph{Necessity.}
If  $n_u\geq k$ and $\dot x = A_u x$ is not $k$-contractive, then, there is a sum of $k$ eigenvalues in the spectrum of $A_u$ that is positive, see Lemma~\ref{lem:spectral-k-properties}. Therefore, since the spectrum of $A_u$ is invariant to the controller gain, the closed-loop system $A-BK$ cannot be $k$-contractive
(invoking again Lemma~\ref{lem:spectral-k-properties}), which proves necessity.

\subsection{Inertia theorems generalizing stabilizability conditions}

We state  a set of
new technical lemmas related to the feasibility and the inertia of the 
generalized stabilizability-like inequality
 \eqref{eqn:k_contraction_LMI_linear_stab2+}.
Note that the following lemmas do not require
the pair $(A,B)$ to be controllable, 
contrarily to \cite{wimmer_1976}.

\begin{lemma}\label{lem:splitting++}
Consider a pair of matrices $(A,B)$ and its
canonical decomposition 
\eqref{eqn:controllable_decomposition}. 
Suppose that for some $\mu\in\RR$,
$\cine(A_u-\mu I) = (n_u-\varrho, 0, \varrho)$
with
$\varrho\in \{0,\ldots, n_u\}$.
Then, there
exists a symmetric matrix $W\in \RR^{n\times n}$, with inertia $\cine(W)=(\varrho,0,n-\varrho)$, 
satisfying 
\begin{equation}
        WA^\top  +  AW  -BB^\top    \prec 2\mu W.\label{eqn:eigenvalue_shifting_LMI}
\end{equation}
\end{lemma}
\begin{proof}
Without loss of generality, suppose that the pair $(A,B)$ is in the form \eqref{eqn:controllable_decomposition}. 
Moreover, define the shifted matrices $\hat{A}_{c}\defeq A_{c}-\mu I$ and $\hat{A}_{u}\defeq A_{u}- \mu I$.
Then, since $\cine(\hat{A}_{u}) = (n_u-\varrho, 0, \varrho)$ and by Lemma~\ref{lem:inertia_shifting}, there exist some symmetric matrix $W_u\in\RR^{n_u\times n_u}$ with inertia $\cine(W_u)=(\varrho,0,n_u-\varrho)$ such that
\begin{equation}\label{eqn:uncontrollable_part_LMI}
 W_u \hat{A}_{u}^\top + \hat{A}_{u}W_u  = - Q,
\end{equation}
with $Q\succ 0$. Furthermore, since the pair $(A_c,B_c)$ is controllable, the pair $(-\gamma I-A_c,B_c)$ is also controllable for any $\gamma\in\RR$.
Hence, for $\gamma>0$
large enough and
from the Lypaunov test of controllability \cite[Theorem 12.4]{Hespanha09}, there exist some positive symmetric matrix $W_c\succ 0$
$$
 W_c (-\gamma I-\hat{A}_{c})^\top + (-\gamma I-\hat{A}_{c})W_c  = - B_c B_c^\top
$$
which implies
\begin{equation}\label{eqn:controllable_part_LMI}
 W_c \hat{A}_{c}^\top + \hat{A}_{c}W_c - B_c B_c^\top = -2\gamma W_c.
\end{equation}
With this in mind, consider a symmetric matrix $W$ with inertia $\cine(W)=(\varrho,0,n-\varrho)$  of the form \begin{equation*} 
    W = \begin{bmatrix}
        W_c & 0\\ 0 & \kappa W_u
    \end{bmatrix}   
    \end{equation*}
where $\kappa>0$ has to be fixed, with $W_u$ and $W_c$ satisfying \eqref{eqn:uncontrollable_part_LMI} and \eqref{eqn:controllable_part_LMI}.
Now, by subtracting  $2\mu W$ from the left-hand side of \eqref{eqn:eigenvalue_shifting_LMI}, we get the following equality
\begin{multline}\label{eqn:hermitian_1}
  \begin{bmatrix}
        W_c & 0\\ 0 & \kappa W_u
    \end{bmatrix}   \begin{bmatrix}
            \hat{A}_{c}^\top & A_{12}^\top\\ 0 & \hat{A}_{u}^\top
        \end{bmatrix}
        +
        \begin{bmatrix}
            \hat{A}_{c} & A_{12}\\ 0 & \hat{A}_{u}
        \end{bmatrix}\begin{bmatrix}
        W_c & 0\\ 0 & \kappa W_u
    \end{bmatrix} 
    \\- \begin{bmatrix}
        B_c B_c^\top & 0 \\0 & 0
    \end{bmatrix} 
    = \begin{bmatrix}
        -2\gamma W_c & \kappa A_{12} W_u \\ \kappa W_uA_{12} ^\top & 
        -\kappa Q
    \end{bmatrix}.
\end{multline}
Since $W_c$ and $Q$ are positive definite, 
the right hand side of identity 
\eqref{eqn:hermitian_1}  can be made negative definite by taking $\kappa>0$  sufficiently small, see e.g. \cite[Section 14.4]{Hespanha09}.
\end{proof}

Additionally, we present the following technical lemma.
\begin{lemma}\label{lem:ricatti_spectrum_separation}
Consider a pair of matrices $(A,B)$ and its
canonical decomposition 
\eqref{eqn:controllable_decomposition}. Moreover, assume there exists a (non-singular) symmetric matrix $W\in\RR^{n\times n}$ and a constant $\mu\in\RR$ such that
\begin{equation}\label{eqn:ricatti_spectrum_separation}
         WA^\top  +  AW  -BB^\top    \prec 2\mu W.
    \end{equation}
    Then, $\pi_{-}({W})\geq \pi_{-}(-A_u+\mu I)$.    
\end{lemma}
\begin{proof}
Without loss of generality, we suppose that the pair $(A,B)$ is in the form \eqref{eqn:controllable_decomposition}. Moreover, notice that the inequality \eqref{eqn:ricatti_spectrum_separation} can be re-arranged as follows 
\begin{equation}\label{eqn:ricatti_spectrum_separation2}
 W\bar A^\top  +  \bar A W \prec 2\mu W,
\end{equation}
where $\bar A := A - \dfrac{1}{2}BB^\top W^{-1}$.
Now, recall the notation in \eqref{eqn:controllable_decomposition} and notice that the eigenvalues in the spectrum of $A_u$ cannot be modified by the term $\dfrac{1}{2}BB^\top P^{-1}$, thus, we have $\sigma(A_u)\subsetneq \sigma(\bar A)$, or, equivalently, $\sigma(A_u-\mu I)\subsetneq \sigma(\bar A-\mu I)$. From this fact we obtain that $\pi_{-}(-\bar A+\mu I)\geq \pi_{-}(-A_u+\mu I)$.  Then, by Lemma~\ref{lem:inertia_shifting} we have that any (non-singular) $W$ that satisfies \eqref{eqn:ricatti_spectrum_separation2} necessarily implies $\pi_{-}(-\bar A+\mu I) = \pi_{-}(W)$, which concludes the proof.
\end{proof}

Finally, we present a technical lemma that relates the colinearity condition in \eqref{eqn:colinear} and the 
 stabilizability-like inequality
 \eqref{eqn:k_contraction_LMI_linear_stab2+}.
\begin{lemma}\label{lem:colinear_ricatti}
Consider a pair of matrices $(A,B)$ and its
canonical decomposition 
\eqref{eqn:controllable_decomposition}. 
Suppose that for some $\mu_1,\mu_2\in\RR$ with $\mu_1\leq\mu_2$,
$\cine(A_u-\mu_1 I) = (n_u-\varrho_1, 0, \varrho_1),\cine(A_u-\mu_2 I) = (n_u-\varrho_2, 0, \varrho_2)$
with
$\varrho_1,\varrho_2\in \{0,\ldots, n_u\}$.
Then, there
exist a pair of symmetric matrices $W_1,W_2\in \RR^{n\times n}$, with inertia $\cine(W_1)=(\varrho_1,0,n-\varrho_1),\cine(W_2)=(\varrho_2,0,n-\varrho_2)$, 
satisfying 
\begin{subequations}\label{eqn:Colinear_Ricatti}
\begin{align}
   &W_1A^\top  +  AW_1  -BB^\top    \prec 2\mu_1 W_1, \label{eqn:W1_ricatti} \\
        &W_2A^\top  +  AW_2  -BB^\top \prec 2\mu_2 W_2, \label{eqn:W2_ricatti}
\end{align}
    \end{subequations}
\begin{equation}\label{eqn:colinearity_W1W2}
  B^\top W_2^{-1} = B^\top W_1^{-1}\,.
\end{equation}
\end{lemma}

\begin{proof}
Without loss of generality, we suppose  the pair $(A,B)$ to be in the form \eqref{eqn:controllable_decomposition}. This is not a restrictive assumption since the colinearity condition \eqref{eqn:colinearity_W1W2} is preserved under linear coordinate changes $z = Tx$, for any non-singular constant matrix $T\in\RR^{n\times n}$.  

Now, since $\cine(A_u-\mu_1 I) = (n_u-\varrho_1, 0, \varrho_1)$, we can follow similar arguments as in Lemma~\ref{lem:splitting++} proof, to show that \eqref{eqn:W1_ricatti} can be satisfied with a symmetric matrix $W_1$ of the form 
\begin{equation}\label{eqn:W1_candidate} 
    W_1 = \begin{bmatrix}
        W_c & 0\\ 0 & \kappa_1 W_{1,u}
    \end{bmatrix}   
    \end{equation}
where $\kappa_1>0$ is a sufficiently small constant, $W_{1,u}\in\RR^{n_u\times n_u}$ is a symmetric matrix with inertia  $\cine(W_{1,u})=(\varrho_1,0,n_u-\varrho_1)$ and $W_c\succ 0$ is a positive definite symmetric matrix computed from
\begin{equation}\label{eqn:controllable_part_W1}
 W_c (A_c-\mu_1 I)^\top + (A_c-\mu_1 I)W_c - B_c B_c^\top = -2\gamma W_c.
\end{equation}
for some positive $\gamma>0$.

With this in mind, we can construct a solution to \eqref{eqn:W2_ricatti} such that $W_1$ and $W_2$ are colinear according to \eqref{eqn:colinearity_W1W2}. Since  $\cine(A_u-\mu_2 I) = (n_u-\varrho_2, 0, \varrho_2)$ and by Lemma~\ref{lem:inertia_shifting}, there exists some symmetric matrix $W_{2,u}\in\RR^{n_u\times n_u}$ with inertia $\cine(W_{2,u})=(\varrho_2,0,n_u-\varrho_2)$ such that, for some $Q\succ 0$,
\begin{equation}\label{eqn:uncontrollable_part_W2}
 W_{2,u} (A_u-\mu_2 I)^\top + (A_u-\mu_2 I)W_{2,u}  = - Q.
\end{equation}
Moreover, recall the relation \eqref{eqn:controllable_part_W1}, then, we derive the following
\begin{equation}\label{eqn:W2_Wc_relation}
\begin{aligned}
&W_c (A_c-\mu_2)^\top + (A_c-\mu_2 I)W_c- B_c B_c^\top 
\\
&=W_c (A_c-\mu_1 I)^\top + (A_c-\mu_1 I)W_c + 2(\mu_1-\mu_2)W_c - B_c B_c^\top \\
&=-2(\gamma+\mu_2-\mu_1)W_c.
\end{aligned}
\end{equation}
With this in mind, consider a symmetric matrix $W_2$  of the form \begin{equation}\label{eqn:W2_candidate} 
    W_2 = \begin{bmatrix}
        W_c & 0\\ 0 & \kappa_2 W_{2,u}
    \end{bmatrix}   
    \end{equation}
where $\kappa_2>0$ has to be fixed, with $W_{2,u}$ and $W_c$ satisfying \eqref{eqn:uncontrollable_part_W2} and \eqref{eqn:W2_Wc_relation}.
Now, by subtracting  $2\mu_2 W_2$ from the left-hand side of \eqref{eqn:W2_ricatti} and defining $\hat{A}_{c}:= A_c-\mu_2 I$ and $\hat{A}_{u}:= A_u-\mu_2 I$, we get the following equality
\begin{multline}\label{eqn:hermitian_2}
  \begin{bmatrix}
        W_c & 0\\ 0 & \kappa W_{2,u}
    \end{bmatrix}   \begin{bmatrix}
            \hat{A}_{c}^\top  & A_{12}^\top\\ 0 & \hat{A}_{u}^\top
        \end{bmatrix}
        +
        \begin{bmatrix}
            \hat{A}_{c} & A_{12}\\ 0 & \hat{A}_{u}
        \end{bmatrix}\begin{bmatrix}
        W_c & 0\\ 0 & \kappa W_{2,u}
    \end{bmatrix} 
    \\- \begin{bmatrix}
        B_c B_c^\top & 0 \\0 & 0
    \end{bmatrix} 
    = \begin{bmatrix}
        -2(\gamma+\mu_2-\mu_1)W_c & \kappa_2 A_{12} W_{2,u} \\ \kappa_2 W_{2,u}A_{12} ^\top & 
        -\kappa_2 Q
    \end{bmatrix}.
\end{multline}
Recall that $\mu_1\leq\mu_2$ by assumption and $\gamma>0,W_c\succ 0$ by design. Consequently, $-2(\gamma+\mu_2-\mu_1)W_c$ is negative definite and
the right hand side of identity 
\eqref{eqn:hermitian_2}  can be made negative definite by taking $\kappa_2>0$  sufficiently small. 

Finally, since the system is in the form \eqref{eqn:controllable_decomposition}, we have that $W_1,W_2$ constructed as in the block-diagonal form \eqref{eqn:W1_candidate} and \eqref{eqn:W2_candidate} satisfy \eqref{eqn:colinearity_W1W2}.  
\end{proof}
\subsection{Proof of Theorem~\ref{thm:k-order-stabilizability}}\label{sec:k-order-stabilizability_proof}

Now, similar to the proof in Section~\ref{subsec:proof_theorem_cLTI},  we introduce a notation in order to represent  the eigenvalues of $A_u$ and their associated multiplicities. Precisely,
consider the matrix $A_u$ in  \eqref{eqn:controllable_decomposition} 
 and $\sigma(A_u)$ its spectrum  and suppose
\(
\Pi(\sigma(A_u))
= \{\alpha_1, \alpha_2, \dots, \alpha_q\}
\)
($q\le n_u$)
where again  $\Pi:\CC\to\RR$ denote the canonical projection onto the real axis and
with 
\(
\alpha_1 > \alpha_2 > \dots > \alpha_q.
\)
Set
\(
\bar h_i = \card \big(\Pi^{-1}(\alpha_{i+1}) \bigcap \sigma(A_u)\big),
\)
where eigenvalues have been counted with their algebraic multiplicities
(so that $\bar h_1+ \bar h_2 + \cdots + \bar h_{q-1} = n_u$). Finally, let $\bar d_0=0$,
\(
\bar d_i =\sum_{j=1}^{i-1} \bar h_j, \; i\in\{1, \dots, {q-1}\}
\). Similar to Section~\ref{subsec:proof_theorem_cLTI}, $\bar h_i$ represents the number of eigenvalues of $A_u$ projected to $\alpha_{i+1}$ and $\bar d_i$ represents the amount of eigenvalues with real part strictly than $\alpha_{i+1}$. 

Finally,  by recalling Lemma~\ref{lem:k-order_stabilizability} 
and Lemma~\ref{lem:spectral-k-properties},  a necessary and sufficient condition for $k$-order stabilizability is either $n_u<k$ or 
\begin{equation}\label{eqn:necessary_cond_example_stab}
    \sum_{j =1}^{k}\Re{(\lambda^u_{j})}<0.
\end{equation} 

We now present the main arguments proving that the inequalities \eqref{eqn:k_contraction_LMI_linear_stab2_total} are necessary and sufficient for $k$-order stabilizability.

\smallskip
 \noindent
\emph{Sufficiency.}  
The goal of this proof is to show
that if \eqref{eqn:k_contraction_LMI_linear_stab2_total} is satisfied, then either \eqref{eqn:necessary_cond_example_stab} is satisfied or $n_u<k$, hence showing the result
invoking
Lemmas~\ref{lem:k-order_stabilizability} and \ref{lem:spectral-k-properties}. Without loss of generality, we suppose that the pair $(A,B)$ is in the form \eqref{eqn:controllable_decomposition}.

Firstly, we assume the case $k\leq n_u$ and we show that \eqref{eqn:k_contraction_LMI_linear_stab2_total} implies \eqref{eqn:necessary_cond_example_stab}. To this end and by means of Lemma ~\ref{lem:ricatti_spectrum_separation}, we have that inequality \eqref{eqn:k_contraction_LMI_linear_stab2+} implies that $\pi_{-}(W_i)\geq\pi_{-}(-A_u+\mu_i I)$ for  all $i \in \{0,\dots,\ell-1\}$. Recalling that $W_i$ has inertia $\cine(W_i)=(d_i,0,n-d_i)$,
we have that $A_{u}$ has at most $d_i$ eigenvalues with real part strictly larger than $\mu_i$.  
This is equivalent to the bound 
\begin{equation}\label{eqn:eig_bound_stab}
    \Re(\lambda^u_{d_{i+1}})\leq \Re(\lambda^u_{d_{i}+1})<\mu_i, \qquad \forall i\in \{0, \dots, \ell-1\}.
\end{equation}

Then, following similar arguments as in the sufficiency part of Section~\ref{subsec:proof_theorem_cLTI}, the next bound can be obtained.
\begin{equation}
    \sum_{i =1}^{k}\Re{(\lambda^u_{i})} < \sum_{i=0}^{\ell-1}{h_{i}} \, \mu_i\le 0\, .
\end{equation}

Thus, \eqref{eqn:necessary_cond_example_stab} is satisfied and the pair $(A,B)$ is $k$-order stabilizable invoking
Lemmas~\ref{lem:k-order_stabilizability} and \ref{lem:spectral-k-properties}.

Finally, for the case $k>n_u$, we have $k$-order stabilizability directly from Lemma~\ref{lem:k-order_stabilizability}, thus ending the sufficiency proof.

 \noindent
\emph{Necessity.}  As stated before,  if the pair $(A,B)$ is $k$-order stabilizable then, either \eqref{eqn:necessary_cond_example_stab} or $n_u<k$ is verified. The goal of this proof is to show that if one of these conditions are satisfied, then, there exists a solution to the inequalities \eqref{eqn:k_contraction_LMI_linear_stab2_total}.
We begin by assuming the case $k\leq n_u$ and \eqref{eqn:necessary_cond_example_stab} is verified. Now, let the scalars $p_k$ and $c_k$ be defined as in \eqref{eqn:allowed_pdom_indices}.

Then, the following equality holds
\begin{equation}\label{eqn:alpha_equality_stab}
(k-p_k)\alpha_{c_k} +\sum_{i=0}^{c_k-2}{h_{i}} \, \alpha_{i+1} =\sum_{j =1}^{k}\Re{(\lambda^u_{j})}.
\end{equation}
Hence, combining \eqref{eqn:necessary_cond_example_stab} and \eqref{eqn:alpha_equality_stab}, if the system is $k$-order stabilizable (with $k\leq n_u$), the next bound is satisfied 
\begin{equation}\label{eqn:necessary_cond_alpha_stab}
(k-p_k)\alpha_{c_k} +\sum_{i=0}^{c_k-2}{h_{i}} \, \alpha_{i+1} < 0.
\end{equation}
Then, by continuity, there exist a scalar $\varepsilon>0$, such that
$$
    (k-p_k) (\alpha_{c_k}+\varepsilon) +\sum_{i=0}^{c_k-2}{h_{i}} \, (\alpha_{i+1}+\varepsilon) \le 0\,
$$ 
Now, fix $\ell = c_k$, $d_\ell = k-p_k+d_{\ell-1}$, $d_i = \bar d_i$ for all $i\in\{0,\dots, \ell-1\}$ and select
$
\mu_{i-1}=\varepsilon + \alpha_{i},$ for all $ i\in\{1,\dots,\ell\},
$. We have 
\begin{align*}
\sum_{i=0}^{\ell-1}{h_{i}}\mu_i&=(k-p_k) (\alpha_{c_k}+\varepsilon) +\sum_{i=0}^{c_k-2}{h_{i}} \, (\alpha_{i+1}+\varepsilon) \le 0,
\end{align*}
thus showing \eqref{eqn:mu_sum_condition_stab+}. Now, since $\mu_{i}>\alpha_{i+1}$ for all $i\in\{0,\dots, \ell-1\}$ we have that $A_u$ has only $\bar d_i$ eigenvalues strictly larger than $\mu_{i}$ and the rest are strictly smaller. That is, $\cine(A_u-\mu_i I) = (n_u-\bar d_i, 0, \bar d_i)$ for all $i\in\{0,\ell-1\}$. Then, by Lemma~\ref{lem:splitting++}, there exist symmetric matrices $W_i$ with $\cine(W_i)=\{\bar d_i,0,n-\bar d_i\}$ such that 
$$
 A^\top W_i +  W_iA -BB^\top \prec 2\mu_i W_i\qquad \forall i\in\{0,\dots, \ell-1\}\,,
$$
thus concluding the  proof if \eqref{eqn:necessary_cond_example_stab} is verified and $k\leq n_u$.

We now proceed with the necessity proof for the case $k>n_u$.  For this proof, we remark that since a $k$-contractive system is also $\bar k$-contractive for all $\bar k\in \{k,\dots,n\}$ \cite{Wu2022}, we have that if a pair $(A,B)$ is $k$-order stabilizable, then, it is also $\bar k$-order stabilizable for all $\bar k\in \{k,\dots,n\}$. Additionally, by means of Lemma~\ref{lem:k-order_stabilizability}, a pair $(A,B)$ is always $k$-order stabilizable if $k=n_u+1$. Therefore, for all $k>n_u$,  $k$-order stabilizability necessarily implies $\bar k$-order stabilizability with $\bar k = n_u+1$. With this fact in mind, this proof is based on showing that, if $k=n_u+1$, then, there always exists a pair of matrices $W_0,W_1$ and constants $\mu_0,\mu_1$ such that \eqref{eqn:k_contraction_LMI_linear_stab2_total} is satisfied. We highlight that this result does not require \eqref{eqn:necessary_cond_example_stab} to be satisfied.

Precisely, assume $k=n_u+1$. Notice that we can always guarantee $\cine(A_u-\mu_0 I) = (n_u,0,0)$ for any $\mu_0\in \RR$ large enough. Therefore, by considering this sufficiently large $\mu_0$ and by means of Lemma~\ref{lem:splitting++}, we know that there exists a symmetric matrix $W_0$ with inertia $\cine(W_0) = (0,0,n)$ solution of \eqref{eqn:k_contraction_LMI_linear_stab2+}. Furthermore, we can always find a sufficiently negative constant $\mu_1<0$, such that 
\begin{equation}\label{eqn:mu_selection}
    \mu_1+n_u \mu_0\leq 0,
\end{equation}
and $\cine(A_u-\mu_1 I)= (0,0,n_u)$. Therefore, by considering this $\mu_1$ and by means of Lemma~\ref{lem:splitting++}, there exists a symmetric matrix $W_1$ with inertia $\cine(W_1) = (n_u ,0,n-n_u)$ solution of \eqref{eqn:k_contraction_LMI_linear_stab2+}. Finally, fix $\ell = 2$ and select the aforementioned pair of matrices $W_0,W_1$ and pair of constants $\mu_0,\mu_1$ (these matrices and constants satisfy \eqref{eqn:k_contraction_LMI_linear_stab2+}). Moreover, fix $d_\ell = n_u +1$. With this selection, we have $d_0= 0, d_1 = n_u$ and $h_0 = n_u, h_1 = 1$. Thus, \eqref{eqn:mu_selection} implies \eqref{eqn:mu_sum_condition_stab+}.

\subsection{Proof of Proposition~\ref{pro:k-order-feedback}}\label{sec:proof_feedback}
The first part of the proof focuses on proving the existence of solutions for the inequalities \eqref{eqn:k_contraction_LMI_linear_stab2_total} considering the assumptions stated in the theorem and in particular the colinearity condition in \eqref{eqn:colinear}. An immediate result of Lemma~\ref{lem:colinear_ricatti}
is that there always exist a set of constants $\mu_i$ and $W_i$ such that \eqref{eqn:mu_sum_condition_stab+} and the colinearity condition in  \eqref{eqn:colinear} is simultaneously satisfied for all $i\in\{1, \dots,\ell-1\}$. Moreover, notice that Lemma~\ref{lem:colinear_ricatti} preserves the relation between the inertia of $A_u-\mu_i I$ and $W_i$ as in Lemma~\ref{lem:splitting++}. Consequently, the arguments in the necessity part of Section~\ref{sec:k-order-stabilizability_proof} could be repeated to obtain the existence of a solution from a $k$-order stabilizability assumption.

The second part of the proof consist in showing how the state-feedback law \eqref{eqn:feedback_design+} makes the closed-loop system
\begin{equation}\label{eqn:closed_loop_sys}
\dot x = (A-BK)x = (A-\frac\rho2 B B^\top  W_0^{-1})x,
\end{equation}
 $k$-contractive for all $\rho\geq1$. Note that, since $W_i$ is non-singular and symmetric for all $i\in \{0, \dots, \ell-1\}$ and by means of the colinearity condition \eqref{eqn:colinear},  the left-hand side of \eqref{eqn:k_contraction_LMI_linear_stab2+} can be rearranged as follows for all $i\in \{0, \dots, \ell-1\}$
$$
\begin{aligned}
&W_iA^\top  +  AW_i  -BB^\top  \\
&\qquad = W_i(A-\tfrac{1}2 BB^\top  W_i^{-1})^\top  +  (A-\tfrac{1}2 BB^\top  W_i^{-1})W_i \\
&\qquad = W_i(A-\tfrac{1}2 BB^\top W_0^{-1})^\top  +  (A-\tfrac{1}2 BB^\top W_0^{-1})W_i.
\end{aligned}
$$
Combining this result with the right-hand side of \eqref{eqn:k_contraction_LMI_linear_stab2+}, we obtain that for all $i\in \{0, \dots, \ell-1\}$
\begin{equation}\label{eqn:gen_Lyap_Wi}
W_i(A-\tfrac{1}2 BB^\top W_0^{-1})^\top  +  (A-\tfrac{1}2 BB^\top W_0^{-1})W_i \prec 2\mu_i W_i.
\end{equation}
Now, by adding $(1-\rho)BB^\top$ in both sides of \eqref{eqn:gen_Lyap_Wi} 
and  considering the fact that $(1-\rho)BB^\top\preceq 0$ for all $\rho\geq 1$, by  \eqref{eqn:colinear}  we get,
\begin{multline}\label{eqn:gen_Lyap_Wi2}
W_i(A-\dfrac{\rho}2 BB^\top W_0^{-1})^\top  +  (A-\dfrac{\rho}2 BB^\top W_0^{-1})W_i  \\
\prec 2\mu_i W_i + (1-\rho)BB^\top \preceq2\mu_i W_i.
\end{multline}
By post-multiplying and pre-multiplying both side of \eqref{eqn:gen_Lyap_Wi2} by $W_i^{-1}$ and defining $P_i:=W_i^{-1}$ we get for all $i\in \{0, \dots, \ell-1\}$
\begin{equation}\label{eqn:gen_Lyap_Pi}
(A-\dfrac{\rho}2 BB^\top W_0^{-1})^\top P_i  +  P_i(A-\dfrac{\rho}2 BB^\top W_0^{-1}) \prec 2\mu_i P_i.
\end{equation}
Finally, combining \eqref{eqn:gen_Lyap_Pi} and \eqref{eqn:mu_sum_condition_stab+} with Theorem~\ref{thm:k_contraction_LMI_cLTI} proves that the closed-loop system \eqref{eqn:closed_loop_sys} is $k$-contractive.

\section{Proofs of  nonlinear results}\label{sec:proofs_NL}
\subsection{Preliminary results}
We provide in this section 
some preliminary results that will be used 
in the proof of Theorem~\ref{thm:k_contraction_LMI}.
We start by recalling the definition of $p$-dominance \cite{Forni2019}.  
\begin{definition}[$p$-dominance]\label{def:p-dominance}
    System \eqref{eqn:original_system} is said to be strictly $p$-dominant on $\cS \subsetneq\RR^n$ if\footnote{The definition can be extended to the full set $\RR^n$ but in this case condition \eqref{eqn:p_dominance_LMI} is modified into
$P \tfrac{\partial f}{\partial x}(x) 
+ 
\tfrac{\partial f}{\partial x}(x)^\top P     \preceq -2\mu P - \varepsilon I$    for all $x\in \RR^{n}$ where the term 
$-\varepsilon I$
is added to ensure uniformity.} there exist a real number $\mu\geq0$ and a symmetric matrix $P\in\RR^{n\times n}$ with inertia $\cine(P)=(p,0,n-p)$ such that
    \begin{equation}\label{eqn:p_dominance_LMI}
P \dfrac{\partial f}{\partial x}(x) 
+ 
\dfrac{\partial f}{\partial x}(x)^{\!\top} P     \prec -2\mu P\,, 
    \quad\forall x\in \cS.
    \end{equation}
\end{definition} 

Then, we recall (with a mild reformulation) the following result on $p$-dominance \cite[Theorem 1]{Forni2019}.

\begin{theorem}\label{thm:Magic_splitting}
  Suppose that system \eqref{eqn:original_system} is
strictly $p$-dominant on 
a compact forward invariant set 
     $\cA \subsetneq \RR^n$  with rate $\mu > 0$ and symmetric matrix $P$ with inertia $\cine(P)=(p,0,n-p)$. Then, for each $x\in \cA$, there exists an invariant splitting $T_x \RR^n = \cV_x \oplus \cH_x$, i.e. there exists a continuous mapping $\transf:\RR^n\to\RR^{n\times n}$ 
invertible for any $x\in \cA$ and
satisfying
\begin{subequations}
\begin{equation}\label{eqn:matrix_splitting}
\transf(x) \defeq \begin{bmatrix}
    \transfh(x) & \transfv(x)
\end{bmatrix}\,,
\end{equation}
where $\transfh: \RR^n\to\RR^{n\times n-p}$ and  $\transfv: \RR^n\to\RR^{n\times p}$ satisfy
\begin{equation}\label{eqn:splitting}
\Imag \ \transfh(x) = \cH_x, \quad \Imag \ \transfv(x) = \cV_x.
\end{equation} Moreover, there exist a scalar $c_h>0$ such that 
\begin{equation}\label{eqn:Horizontal_contraction_bound}
\left|\difflow^t(x)\begin{bmatrix}
    \transfh(x) & 0
\end{bmatrix}\diffvar\right| \leq c_h e^{-\mu t}\left|\begin{bmatrix}
    \transfh(x) & 0
\end{bmatrix}\diffvar\right|
\end{equation}
\end{subequations}
holds for all $ (t,x, \diffvar) \in\RR_{\geq0}\times\cA\times T_x \RR^n$.
\end{theorem}

With this in mind, it is clear that if $\mu_1$ is strictly negative, the matrix inequality \eqref{eqn:Pk_LMI} imposes a form of horizontal contraction on the system \cite[Section VII]{Forni2014}. Nonetheless, horizontal contraction is not a sufficient condition for $k$-contraction \cite{Wu2022b}. 
This motivates \eqref{eqn:P0_LMI}. We clarify the effects of \eqref{eqn:P0_LMI} via the following Lemma.

\begin{lemma}\label{lemma:expansion}
Consider system \eqref{eqn:original_system} and assume there exist a forward invariant compact set $\cA\subsetneq \RR^n$, a symmetric positive definite matrix $P_{0}\in\RR^{n\times n}$  and a scalar $\mu_0$ satisfying \eqref{eqn:P0_LMI} for all $x\in\cA$. Then there exists a constant $c_v>0$ 
such that
\begin{equation}\label{eqn:Vertical_contraction_bound}
\left|\difflow^t(x)\begin{bmatrix}
    0 & \transfv(x)
\end{bmatrix}\diffvar\right| \leq c_v e^{\mu_0 t}\left|\begin{bmatrix}
    0 & \transfv(x)
\end{bmatrix}\diffvar\right|
    \end{equation}
    for all $ (t,x, \diffvar ) \in\RR_{\geq0}\times\cA\times T_x \RR^n$, 
     with $\transfv$ as in \eqref{eqn:splitting}. 
\end{lemma}

\begin{proof}
    Consider the function,
   $W \defeq \diffvar^\top P_0\diffvar$.
    It satisfies
\begin{equation}\label{eqn:WQ_bounds}
        \underline{\lambda}(P_0)|\diffvar|^2\leq W(v) \leq \overline{\lambda}(P_0)|\diffvar|^2,
    \end{equation}
    where $ \underline{\lambda}(\cdot)$ and $\overline{\lambda}(\cdot)$ represent the minimum and maximum eigenvalue of their argument, respectively.
    By \eqref{eqn:variational_system}, its time-derivative  satisfies
    \begin{align*}
        \dot{W} &= \diffvar^\top\left(P_0\dfrac{\partial f}{\partial x}(x) + \dfrac{\partial f}{\partial x}(x)^\top P_0 \right)\diffvar \\
        &< 2\mu_0\diffvar^\top P_0 \diffvar = 2\mu_0W.
    \end{align*}
    Then, by Gr\"onwall–Bellman inequality, we obtain 
    \begin{equation*}
        W(t) \leq 
        e^{2\mu_0t}W(0), 
        \quad \forall t\in \RR_{\geq0}.
    \end{equation*}
    Invoking \eqref{eqn:WQ_bounds}, we obtain 
    for all $ (t,x, \diffvar) \in\RR_{\geq0}\times\cA\times T_x \RR^n$ 
    \begin{equation*}
        \big|\difflow^t(x)\diffvar\big| \leq\sqrt{\dfrac{\overline{\lambda}(P_0)}{\underline{\lambda}(P_0)}}e^{\mu_0t}|\diffvar|.
    \end{equation*}
    As $\begin{bmatrix}
    0 & \transfv(x)
\end{bmatrix}\diffvar\in T_x\RR^n$, the result trivially follows.
\end{proof}

Given the above results, condition \eqref{eqn:mu_sum_condition_NL} can be seen as imposing a bound on the maximum expansion rate of the vertical subspace with respect to the contraction rate of the horizontal one. In particular, \eqref{eqn:mu_sum_condition_NL} holds if the first is smaller than the latter. We now relate this property to infinitesimal $k$-contraction. As a first step, we present a technical lemma related to matrix compounds.
\begin{lemma}\label{lem:technical_lemma_compounds}
    Consider a time-varying matrix $M(t)\in\RR^{n\times n}$ 
    $$
    M(t) = \begin{bmatrix}
    H(t) & V(t)
    \end{bmatrix},
    $$
    with $H(t)\in \RR^{n \times n-p}$, $V(t) \in \RR^{n \times p}$ and $p\in[0,n)$. Assume there exist real numbers $c_h, c_v, \alpha, \beta>0$ such that
\begin{equation}\label{eqn:submatrix_bound}
    |H(t)| \leq c_h e^{-\alpha t}, \quad |V(t)| \leq c_v e^{\beta t},
    \qquad\forall t\in \RR_{+}.
    \end{equation}
    If $\alpha > (k-1)\beta$  for some integer $k\in [p+1, n]$, there exist some real numbers $c,\varepsilon>0$ such that 
    \begin{equation}\label{eqn:M_bound}
    |M(t)^{(k)}| \leq c e^{-\varepsilon t},
    \qquad \forall t\in \RR_{+}.
    \end{equation}
\end{lemma}
\begin{proof}
    Consider the elements of the compound matrix $M(t)^{(k)}$. Each one is a $k^{th}$-order minor of the original matrix $M(t)$, i.e., it is the determinant of a $k\times k$ submatrix of $M(t)$, see Definition~\ref{def:compound_definition}. Since $k\geq p+1$, each $k\times k$ submatrix contains at least one column composed of elements of $H(t)$. That is, in the minimum case
    \begin{equation}\label{eqn:worst_case_submat}
    M_k(t) = \begin{bmatrix}
            h(t) & v_1(t) & \dots & v_{k-1}(t)
            \end{bmatrix},
    \end{equation}
    where $M_k(t)\in\RR^{k\times k}$ is a submatrix of $M(t)$, $h(t)\in\RR^{k}$ is a vector with components of $H(t)$ and $v_i(t)\in\RR^{k}$ for $i=1,\dots,k-1$ is a vector with components of $V(t)$. 
    In what follows, we show the elements of $M(t)^{(k)}$ are bounded. Hence, we focus on submatrices of the form \eqref{eqn:worst_case_submat}, since their determinant represents the worst-case scenario in a stability sense. By definition of the wedge product,
    $$
    \det(M_k(t)) = h(t)\wedge v_1(t) \wedge \dots \wedge v_{k-1}(t).
    $$ The wedge product can be represented using a basis $e^i$, where $e^i$ depicts the $i$th canonical vector of $\RR^n$. More specifically, by bilinearity of the wedge product, we have
    $$
    \begin{aligned}
        \det(M_k(t))
        &= \sum_{i=1}^{n}h^i(t)(e^i\wedge v_1(t)\wedge \dots \wedge v_{k-1}(t)),
    \end{aligned}
    $$
    where $h^i(t)$ is the $i$th element of $h(t)$. By performing similar operations on the remaining vectors we deduce
    \begin{equation}\label{eqn:wedge_opening}
    \begin{aligned}
    \det(M_k(t))
         &=\sum_{i_1=1}^{k} 
         \dots\sum_{i_k=1}^{k}h^{i_1}(t)v_2^{i_2}(t) \dots v_{k-1}^{i_k}(t) E_k,
    \end{aligned}
    \end{equation}
    where $E_k\defeq(e^{i_1} \wedge e^{i_2} \wedge \dots \wedge e^{i_k})$.
    By \eqref{eqn:submatrix_bound}, we have
    $$
    |h^i(t)| \leq c_h e^{-\alpha t}, \quad |v^i(t)| \leq c_v e^{\beta t}.
    $$
    Moreover, the factor $E_k$ will be either zero or an element of the canonical basis in $\RR^n$ multiplied by plus or minus one. Thus, using the triangle inequality, one obtains
    $$
    |\det(M_k(t))| \leq \kappa c_h c_v^{k-1} e^{(-\alpha +(k-1)\beta) t}
    $$
    where $\kappa>0$ is a positive constant related to the number of non-zero instances of $E_k$. Now, since $\alpha-(k-1)\beta>0$ by assumption, by continuity there always exists $\varepsilon>0$ such that $\alpha-(k-1)\beta-\varepsilon>0$. 
    Then,
    $$
    |M(t)^{(k)}| = |e^{-\varepsilon t} e^{\varepsilon t}M(t)^{(k)}| \leq e^{-\varepsilon t} |e^{\varepsilon t}M(t)^{(k)}|.
    $$
    By  considering the worst-case \eqref{eqn:worst_case_submat}, we have 
    $$
    e^{\varepsilon t}|\det(M_k(t))| \leq \bar c e^{(-\alpha +(k-1)\beta + \varepsilon) t},
    $$
    for some $\bar c>0$.
    Hence, since $\alpha-(k-1)\beta-\varepsilon>0$, each element of $e^{\varepsilon t}M(t)^{(k)}$ is exponentially decreasing and the norm
    $|e^{\varepsilon t}M(t)^{(k)}|$ is uniformly bounded for all $t\in\RR_{\geq0}$, thus concluding the proof.
\end{proof}

Leveraging on the previous lemmas, we provide a bound on the  compound of the variational system \eqref{eqn:variational_system} state transition matrix.

\begin{lemma}\label{lemma:transition_comp_contr}
    Consider system \eqref{eqn:original_system} and assume there exist two symmetric matrices $P_0, P_1\in \RR^{n\times n}$ 
 of respective inertia  $(0,0,n)$ and $(k-1,0,n-k+1)$,
  and  $\mu_0,\mu_1\in \RR$ such that 
    \eqref{eqn:2_contraction_LMI} is
    satisfied. Then, there exist  $\varepsilon,c>0$ such that
\begin{equation}\label{eqn:bound_transition_matrix}
    \left|\difflow^t(x)^{(k)}\right| \leq ce^{-\varepsilon t},
    \quad
    \forall (t,x)\in \RR_{\geq0}\times \cA.
    \end{equation}
\end{lemma}
\begin{proof}
   Consider \eqref{eqn:matrix_splitting} in Theorem~\ref{thm:Magic_splitting}. Invertibility of  $\transf(x)$ yields
   $$
   \begin{aligned}
    \difflow^t(x) &= \difflow^t(x)\transf(x)\transf(x)^{-1} = \pmb\psi^t(x)\transf(x)^{-1},
    \end{aligned}
   $$
   with  
   $
   \pmb\psi^t(x) \defeq \begin{bmatrix}
        \difflow^t(x)\transfh(x) & \difflow^t(x)\transfv(x)
    \end{bmatrix}.
   $
  Given any $\diffvar \in T_x\RR^n$, consider the decomposition  
   $\diffvar = (\diffvar^h, \diffvar^v)$, 
   where $\diffvar^h\in\RR^{n-p}$ and $\diffvar^v\in\RR^{p}$. Then, for an arbitrary $\diffvar^h$, inequality \eqref{eqn:Horizontal_contraction_bound} of Theorem~\ref{thm:Magic_splitting} implies
   $$
    |\difflow^t(x)\transfh(x)\diffvar^h| \leq c_h e^{\mu_1}|\transfh(x)\diffvar^h|\, .
   $$
   Recall the definition of matrix norm, 
   \begin{align*} \left|\difflow^t(x)\transfh(x)\right|\defeq\max_{|u|=1} \left|\difflow^t(x)\transfh(x)u\right|\,.
   \end{align*}
   By selecting vector $u^\star$ such that $|u^\star|=1$, the previous exponential relation and the triangular inequality yield
   \begin{align*}
\left|\difflow^t(x)\transfh(x)\right|&=
\left|\difflow^t(x)\transfh(x) u^\star\right|\\
       &\le c_h e^{\mu_1}|\transfh(x)u^\star|
       \le c_h e^{\mu_1}|\transfh(x)|.
   \end{align*}
   Since $\cA$ is compact and $\transf$ is continuous, $|\transfh(x)|$ is bounded for all $x\in\cA$.
   Then, by \eqref{eqn:Horizontal_contraction_bound}, 
   and by \eqref{eqn:Vertical_contraction_bound} 
   we obtain
\begin{align*}
    \left|\difflow^t(x)\transfh(x)\right| &\leq c_he^{\mu_1}|\transfh(x)| \leq \bar c_h e^{-\mu_1}
    \\
    \left|\difflow^t(x)\transfv(x)\right| &\leq c_ve^{\mu_{0}}|\transfv(x)| \leq \bar c_v e^{\mu_{0}}\,
\end{align*}
   for all $x\in\cA$.
   Finally, by boundedness of $\transf(x)$  and Lemma~\ref{lem:technical_lemma_compounds}, 
   $$
    \left|\difflow^t(x)^{(k)}\right| \le 
    |\pmb\psi^t(x)^{(k)}||{\transf(x)^{-1}}^{(k)}|\le c e^{-\varepsilon t}, \quad \forall x\in \cA.
   $$
   concluding the proof.
\end{proof}

\subsection{Proof Theorem~\ref{thm:k_contraction_LMI}}\label{subsec:proof_theorem_k_contraction}
Consider the $k$-th multiplicative compound of matrix $\vertmatnonlin{t}{x_0}$ defined as in Section~\ref{subsec:inifnit_k_contr}. From  the Cauchy-Binet formula \cite[Chapter 1]{fallat2022totally} we get:
\begin{equation*}
\begin{aligned}
    \vertmatnonlin{t}{x_0}^{(k)} &= \begin{bmatrix}\difflow^t (x_0)\diffvarzero^1 & \dots & \difflow^t(x_0)\diffvarzero^k \end{bmatrix}^{(k)} \\
    &= \difflow^t (x_0)^{(k)}\vertmatnonlin{0}{x_0}^{(k)}.
\end{aligned}
\end{equation*}
From
    \eqref{eqn:2_contraction_LMI} and Lemma~\ref{lemma:transition_comp_contr} we obtain  for all $(x_0,t)$ in $\cS\times\RR_{\geq0}$
$$
|(\vertmatnonlin{t}{x_0})^{(k)}| \leq ce^{-\varepsilon t}|(\vertmatnonlin{0}{x_0})^{(k)}|.
$$
Hence, the system is infinitesimally $k$-contractive on $\cS$ and the $k$-contractive property follows from Theorem~\ref{thm:inf_k_contr}.

\subsection{Proof of Lemma~\ref{lem:planar_case}}\label{sec:proof_planar}
Let us decompose the Jacobian of the vector field $f$ as follows
$$
\dfrac{\partial f}{\partial x}(x) = \begin{bmatrix}
    F_s(x) & G_{12}(x)\\G_{21}(x) & F_u(x)
\end{bmatrix}.
$$
Then, according to Theorem~\ref{thm:demidovich}, a sufficient condition for $2$-contraction in a set $\cA$ is:
\begin{equation}\label{eqn:suff_cond_2contraction}
    \dfrac{\partial f}{\partial x}(x)^{[2]} =  F_s(x)+ F_u(x) <0, \quad \forall x\in\cS.
\end{equation}
By subtracting $2\mu_0 P_0$ in both sides of \eqref{eqn:P0_LMI}, the following inequality is obtained
$$
    \begin{psmallmatrix}
    F_s(x)-\mu_0 & G_{12}(x)\\G_{21}(x) & F_u(x)-\mu_0
\end{psmallmatrix}P_0
+ P_0
\begin{psmallmatrix}
    F_s(x)-\mu_0 & G_{21}(x)\\G_{12}(x) & F_u(x)-\mu_0
\end{psmallmatrix}
\prec 0.
$$
Since $P_0$ is positive definite, the previous inequality necessarily implies that $\tfrac{\partial f}{\partial x}(x)-\mu_0I$ is Hurwitz for all $x\in\cA$, and, consequently, its determinant is positive. That is,
\begin{equation}\label{eqn:det_mu0}
0<F_s(x) F_u(x)-\mu_0(F_s(x) + F_u(x))+\mu_0^2-G_{12}(x)G_{21}(x).
\end{equation}
Similarly, by subtracting $2\mu_1P_1$ in both sides of \eqref{eqn:Pk_LMI}, the following inequality is obtained
$$
    \begin{psmallmatrix}
    F_s(x)-\mu_1 & G_{12}(x)\\G_{21}(x) & F_u(x)-\mu_1
\end{psmallmatrix}P_1
 + P_1
\begin{psmallmatrix}
    F_s(x)-\mu_1 & G_{21}(x)\\G_{12}(x) & F_u(x)-\mu_1
\end{psmallmatrix}
\prec 0.
$$
Since $\cine(P_1)= (1,0,1)$, by Lemma~\ref{lem:inertia_shifting}, this inequality necessarily implies that $\cine\left(\tfrac{\partial f}{\partial x}(x)-\mu_1I\right)=(1,0,1)$ for all $x\in\cA$, and, consequently, its determinant is negative. That is,
\begin{equation*}
F_s(x) F_u(x)-\mu_1(F_s(x) + F_u(x))+\mu_1^2-G_{12}(x)G_{21}(x)<0,
\end{equation*}
which can be rearranged as
\begin{equation}\label{eqn:det_mu1}
    F_s(x) F_u(x)-G_{12}(x)G_{21}(x)< \mu_1(F_s(x) + F_u(x))-\mu_1^2.
\end{equation}
Now, combining \eqref{eqn:det_mu0} and \eqref{eqn:det_mu1} we get
\begin{equation}\label{eqn:det_combined}
0<(\mu_1-\mu_0)(F_s(x) + F_u(x))+\mu_0^2-\mu_1^2.
\end{equation}
Then, since $\mu_0>\mu_1$  we have that $\mu_1-\mu_0<0$, which, by combining \eqref{eqn:mu_sum_condition_NL} and the fact that $\mu_0,\mu_1$ are real, implies $\mu_0^2-\mu_1^2<0$. Additionally, the conditions \eqref{eqn:P0_LMI}-\eqref{eqn:mu_sum_condition_NL} imply that $F_s(x)+F_u(x)\neq0$ for all $x\in\RR^n$ by Lemma~\ref{lem:inertia_shifting}. Therefore, by combining \eqref{eqn:det_combined} and \eqref{eqn:mu_sum_condition_NL} we get  \eqref{eqn:suff_cond_2contraction}, which ends the proof.
\subsection{Proof of Proposition~\ref{thm:2_contraction_feedback_design}}\label{sec:2_contractive_design_NL}
Consider inequality \eqref{eqn:P0_LMI},  pre-multiply and post-multiply both sides of the inequality by $P_0^{-1}$ and fix $W_0=P_0^{-1}$. Similarly, pre-multiply and post-multiply both sides of the inequality \eqref{eqn:Pk_LMI} by $P_1^{-1}$ and fix $W_1=P_1^{-1}$. Then, according to Theorem~\ref{thm:k_contraction_LMI}, the closed-loop system will be $k$-contractive if there exist two symmetric matrices $W_0, W_1\in \RR^{n\times n}$ 
 of inertia  $ \cine(W_0)=(0,0,n)$, 
 $\cine(W_1)=(1,0,n-1)$
  and  $\bar \mu_0,\bar \mu_1\in \RR$ such that, 
\begin{subequations}\label{eqn:2_contraction_LMI_CL}
\begin{align}\label{eqn:P0_LMI_CL}
         W_0 \left(\tfrac{\partial f}{\partial x}(x)-BK\right)^\top \!\! + \left(\tfrac{\partial f}{\partial x}(x)-BK\right)  W_0    & \prec 2\bar \mu_0 W_0,
\\ \label{eqn:P1_LMI_CL}
         W_1\left(\tfrac{\partial f}{\partial x}(x)-BK\right)^\top \!\!+ \left(\tfrac{\partial f}{\partial x}(x)-BK\right)  W_1  &\prec 2\bar \mu_1  W_1,
 \\ \label{eqn:mu_sum_condition_NL_CL}
 \bar \mu_{1} + (k-1)\bar \mu_0 & <0,
\end{align}
\end{subequations}
for all $x\in\cS$, where $\cS$ is assumed to be compact and forward invariant.
In this proof, we show that if the inequalities in \eqref{eqn:k_contraction_LMI_NLdesign}-\eqref{eqn:mu_sum_condition_NL_stab} are satisfied and the gain matrix $K$ is designed as in \eqref{eqn:feedback_design_NL}, then, the inequalities in \eqref{eqn:P0_LMI_CL}-\eqref{eqn:mu_sum_condition_NL_CL} are also satisfied. Thus, the closed-loop system is $k$-contractive according to Theorem~\ref{thm:k_contraction_LMI}. To this end, note that the left-hand side of \eqref{eqn:P0_LMI_CL} 
with $K$ fixed as in \eqref{eqn:feedback_design_NL}
can be rewritten as:
\begin{multline}
    \label{eq:ineq_nonlinear_stab_proof}
W_0\tfrac{\partial f}{\partial x}(x)^{\!\top}  +  \tfrac{\partial f}{\partial x}(x)W_0  -BB^\top 
\\ -\tfrac{1}{2}{BB^\top W_1^{-1}W_0} 
 - \tfrac{1}{2}W_0 W_1^{-1} BB^\top 
 \\
 \prec \mu_0 W_0 -\tfrac{1}{2}{BB^\top W_1^{-1}W_0} 
- \tfrac{1}{2}W_0 W_1^{-1} BB^\top ,
\end{multline}
where the right hand side
is obtained employing 
\eqref{eqn:k_contraction_LMI_NL_stab}.
Now, let 
$\bar \omega>0$ such that 
$$
 \left[I-\tfrac{1}{2}BB^\top W_1^{-1}\right]W_0\left[I-\tfrac{1}{2}BB^\top W_1^{-1}\right]^\top \preceq (1+\bar \omega)W_0.
$$
Furthermore, note that we have the identity
\begin{multline*}
     -\tfrac{1}{2}{BB^\top W_1^{-1}W_0} 
- \tfrac{1}{2}W_0 W_1^{-1} BB^\top  = 
\\
 \left[I-\tfrac{1}{2}BB^\top W_1^{-1}\right]W_0\left[I-\tfrac{1}{2}BB^\top W_1^{-1}\right]^\top  - W_0\\
-\tfrac{1}{4}BB^\top W_1^{-1}W_0W_1^{-1}BB^\top .
\end{multline*}
As a consequence, 
by adding and substracting 
$\bar \omega W_0$ from the right hand side of 
\eqref{eq:ineq_nonlinear_stab_proof}
and using the two previous equations, 
we obtain
\begin{multline*}
W_0\dfrac{\partial f}{\partial x}(x)^\top  +  \dfrac{\partial f}{\partial x}(x)W_0  -BB^\top
\\-\tfrac{1}{2}{BB^\top W_1^{-1}W_0}
  - \tfrac{1}{2}W_0 W_{1}^{-1} BB^\top 
 \prec  (\mu_0+\bar \omega)W_0
\end{multline*}
Thus, selecting
$\bar \mu_0 = \mu_0 + \bar \omega$, 
shows 
\eqref{eqn:P0_LMI_CL}.
Then, inequality \eqref{eqn:k_contraction_LMI_NL_stab3} can be re-organized to obtain the inequality \eqref{eqn:P1_LMI_CL} with $\bar\mu_1 = \mu_1$.
Finally, by fixing $\omega=(k-1)\bar \omega$,  we have that \eqref{eqn:mu_sum_condition_NL_stab} implies \eqref{eqn:mu_sum_condition_NL_CL}.

\subsection{Proof of Lemma~\ref{lem:3-contraction}}\label{sec:proof_3_contraction}
Due to Lemma~\ref{lem:inertia_shifting}, a necessary condition for the feasibility of \eqref{eqn:P0_LMI}-\eqref{eqn:mu_sum_condition_NL} for $k=3$ is $\mu_1<\mu_0$. Consequently,  \eqref{eqn:mu_sum_condition_NL} implies $\mu_1<0$. Therefore, the inequalities \eqref{eqn:P0_LMI}-\eqref{eqn:mu_sum_condition_NL} for $k= 3$ in a forward invariant set $\cS$ imply $2$-dominance in $\cS$ as defined in Definition~\ref{def:p-dominance}. Finally, the result follows from \cite[Corollary 1]{Forni2019} and that $\cS\subsetneq \RR^n$.

\subsection{Proof of Lemma~\ref{lem:n-contraction}}\label{sec:proof_n_contraction}
Let assumptions of Theorem~\ref{thm:demidovich} hold for $k=n$ in a set $\cS$. We have that $\pderiv{f}{x}(x)^{[n]}=\trace\left(\pderiv{f}{x}(x)\right)< 0$ for all $x\in\cS$. Therefore, at least one eigenvalue of $\pderiv{f}{x}(x)$ is negative for any $x\in \RR^n$
    and in particular for all $x^\circ$ such that 
    $f(x^\circ)=0$,
    thus showing the result.
\color{black}
\section{Conclusions}\label{sec:conclusions}
We presented new  conditions for 
$k$-contraction 
based on the use of 
generalized Lyapunov matrix inequalities.
The proposed conditions can be checked without using matrix compounds. 
In the linear case, 
they reduce 
the $k$-contraction
analysis
to solving a set of matrix inequalities. 
In the nonlinear context, they extend
the well-known 
Demidovich conditions
based on the Jacobian of the vector field along the flow.
 Finally,  we showed that the proposed
conditions can be used to develop new
tools for $k$-contractive feedback design,  so that to extend existing conditions for 
standard $1$-contraction. 

Future works will focus on
further studying the equivalence between $k$-contraction and infinitesimally $k$-contraction, and 
extending the proposed conditions 
to the context of time-varying 
systems and Riemannian metrics, similar to the context of $1$-contraction, see, e.g. 
\cite{LOHMILLER1998683,andrieu2020characterizations}.
Another topic
is the design of $k$-contractive observers.

\section*{Acknowledgements}
We thank the anonymous reviewers for their insightful comments, 
for pointing out the asymptotic behavior of $n$-contractive systems (detailed in Section~\ref{sec:limitations}) and the differences between dominance analysis and $k$-contraction (Remark~\ref{rem:reviewer_com}). We also thank Pietro Lorenzetti for suggesting us with the synchronverter example.

\appendix
\subsection{Proof of Lemma~\ref{lem:coordinate_invariance}}\label{sec:proof_coordinate_inv}

 The uniformity condition in \eqref{eqn:unif_diff} and the fact that $P$ is positive definite and symmetric imply the existence of  constants $\bar \sigma,\underline \sigma>0$ such that, for all $(r,x)$ in $[0,1]^k\times\cS$ 
\begin{align*}
\vol(\varphi&\circ\Phi)
 =\! \bigintsss_{[0,1]^k} \sqrt{\det\left\{\tfrac{\partial \Phi}{\partial r}(r)^\top \tfrac{\partial \varphi}{\partial x}(x)^\top 
 P\tfrac{\partial \varphi}{\partial x}(x)\tfrac{\partial \Phi}{\partial r}(r)\right\}}\, \dr\\
&\leq \sqrt{\bar \sigma}\bigintsss_{[0,1]^k} \sqrt{\det\left\{\tfrac{\partial \Phi}{\partial r}(r)^\top P \tfrac{\partial \Phi}{\partial r}(r)\right\}}\, \dr\, = \sqrt{\bar \sigma} \ \vol(\Phi)
\end{align*}
and
$\vol(\Phi)\leq \vol(\varphi\circ \Phi)/\sqrt{\underline{\sigma}}$. 
Consequently, 
if the system is $k$-contractive in $\cS$ in the original coordinates, we have 
\begin{align*}
\vol(\varphi\circ \psi^t\circ \Phi)
    &\leq \sqrt{\bar \sigma}\vol(\psi^t\circ \Phi) \leq \dfrac{\sqrt{\bar \sigma}}{\sqrt{\underline \sigma }}b e^{-a t}\, \vol(\varphi\circ\Phi),
\end{align*}
 for all $(r,x)$ in $[0,1]^k\times\cS$, thus showing $k$-contraction.

\subsection{Proof of Theorem~\ref{thm:inf_k_contr}}\label{subsec:proof_infin_contr}
Consider $\Phi\in\cI_k$, where $\cI_k$ is defined in \eqref{eqn:parametriz_set}, satisfying $\Imag(\Phi) \subseteq \cS$.
To simplify notation, let us denote for all $(r,t)$ in $[0,1]^k\times\RR_{\geq0}$
\begin{equation}\label{eqn:gamma_func}
\Gamma(r,t) := \psi^t\circ \Phi(r)\ , \ \Gamma_r(r,t) := \frac{\partial \Gamma}{\partial r}(r,t).
\end{equation}
In words, the factor $\Gamma(r,t)$ depicts the solution of \eqref{eqn:original_system} at time $t$ taking as a initial condition a point in $\Phi$ parametrized by $r$. Then, $\Gamma_r(r,t)$ represents
the effect of infinitesimal variations in $r$ on the solution at time $t$.
For all $(r,t)$ in $[0,1]^k\times\RR_{\geq0}$, we have
$$
\dddt \Gamma(r,t) = f(\Gamma(r,t) ).
$$
Moreover, since $\cS$ is forward invariant and $\Imag(\Phi) \subseteq \cS$, we have $\Gamma(r,t)\in\cS$ for all $(r,t)$ in $[0,1]^k\times\RR_{\geq0}$.
Additionally, by the chain rule, it follows that
the point $\Gamma_r(r,t)$ evolves according to
\begin{align*}
\dddt  \Gamma_r(r,t)&=\frac{\partial^2 \Gamma}{\partial r\partial t}(r,t)= 
\dfrac{\partial f}{\partial x}(\Gamma(r,t) )\Gamma_r(r,t)
\end{align*}
for all $(r,t)$ in $[0,1]^k\times\RR_{\geq0}$.
Since these dynamics are linear, following similar steps as in \cite[Section 2.5]{Wu2022}, we obtain 
\begin{equation} \label{eqn:Gamma_compund_dyn}
\dddt \Gamma_r(r,t)^{(k)} = \dfrac{\partial f}{\partial x}(\Gamma(r,t))^{[k]}\Gamma_r(r,t)^{(k)}.
\end{equation}
Similarly, following \cite[Section 2.5]{Wu2022}, it can be shown that the compound matrix of $\vertmatnonlin{t}{x_0}$ evolves according to the linear dynamics
\begin{equation}\label{eqn:compund_dyn}
\dfrac{d}{dt}\left(\vertmatnonlin{t}{x_0}\right)^{(k)} = \dfrac{\partial f}{\partial x}(\psi^t(x_0))^{[k]}\vertmatnonlin{t}{x_0}^{(k)}.
\end{equation}
By \eqref{eqn:compound_stability}, dynamics \eqref{eqn:compund_dyn} are globally exponentially stable.
Thus, considering \eqref{eqn:Gamma_compund_dyn}, exponential stability and uniformity of \eqref{eqn:compund_dyn} imply 
  \begin{equation*}
        \left|\Gamma_r(r,t)^{(k)}\right| \leq b e^{-a t}\left|\Gamma_r(r,0)^{(k)}\right|,
    \end{equation*}
    for some positive constants $a,b>0$. Now, from the Cauchy-Binet formula \cite[Chapter 1]{fallat2022totally} , the following equality holds
\begin{equation}\label{eqn:det_comp_relation}
\begin{aligned} 
   &\det\Big(\Gamma_r(r,t)^\top P\, \Gamma_r(r,t)\Big)  \\
    & \qquad = \!\!\left(\Gamma_r(r,t)^{(k)}\right)^\top \!\!\!\!P^{(k)}\, \Gamma_r(r,t)^{(k)} \!\! := \!\!v(r,t).
\end{aligned}
\end{equation}
Then,
the volume  $\vol(\cdot)$ of $\psi^t\circ\Phi$ 
computed according to 
\eqref{eqn:k_length}
is
$$
\vol(\psi^t\circ\Phi) = \int_{[0,1]^k}\sqrt{v(r,t)}\dr\,.
$$
Finally, by selecting $P$ in \eqref{eqn:det_comp_relation} as the identity matrix,  we obtain 
\begin{align*}
    \vol(\psi^t\circ\Phi) &= \int_{[0,1]^k} \!\!\left|\Gamma_r(r,t)^{(k)}\right|\, dr
    \le \int_{[0,1]^k} \!\!\!\! b e^{-a t}\left|\Gamma_r(r,0)^{(k)}\right|\, dr\\
    & \le b e^{-a t} \int_{[0,1]^k} \left|\Gamma_r(r,0)^{(k)}\right|\, dr
     \le b e^{-a t} \vol(\Phi) \,.
\end{align*}

\subsection{Proof of Theorem~\ref{thm:demidovich}}\label{subsec:proof_demid}
Consider the state transition matrix $\difflow^t (x_0)$ of the variational system \eqref{eqn:variational_system}. Following \cite[Section 2.5]{Wu2022}, it can be shown that the $k$-th order multiplicative compound matrix of $\difflow^t (x_0)$ evolves according to the following linear dynamics for all $(x_0,t)$ in $\cS\times\RR_{\geq0}$
\begin{equation*}
\dfrac{d}{dt}\left(\difflow^t (x_0)\right)^{(k)} = \dfrac{\partial f}{\partial x}(\psi^t(x_0))^{[k]}\difflow^t (x_0)^{(k)}, \quad \difflow^0 (x_0)=I.
\end{equation*}
Therefore, from \cite[Proposition 5]{Wu2022}, it can be shown that \eqref{eqn:compound_LMI} implies the existence of some positive constants $c,\varepsilon>0$ such that
$$
|\difflow^t (x_0)|\leq ce^{-\varepsilon t}, \quad \forall (x_0,t)\in\cS\in\RR_{\geq0}
$$
The proof concludes by repeating the arguments as in Section~\ref{subsec:proof_theorem_k_contraction}.

\color{black}

%%%%%%%%%%%%%%%%%%%%%%%%%%%%%%%%%%%%%%%%%%%%%%%%%%%%%%%%%%%%%%%%%%%%%%%%%%%%%%%%

\bibliographystyle{IEEEtran}
\bibliography{biblio}

\begin{IEEEbiography}
[{\includegraphics[width=1in,height=1.25in,clip,keepaspectratio]{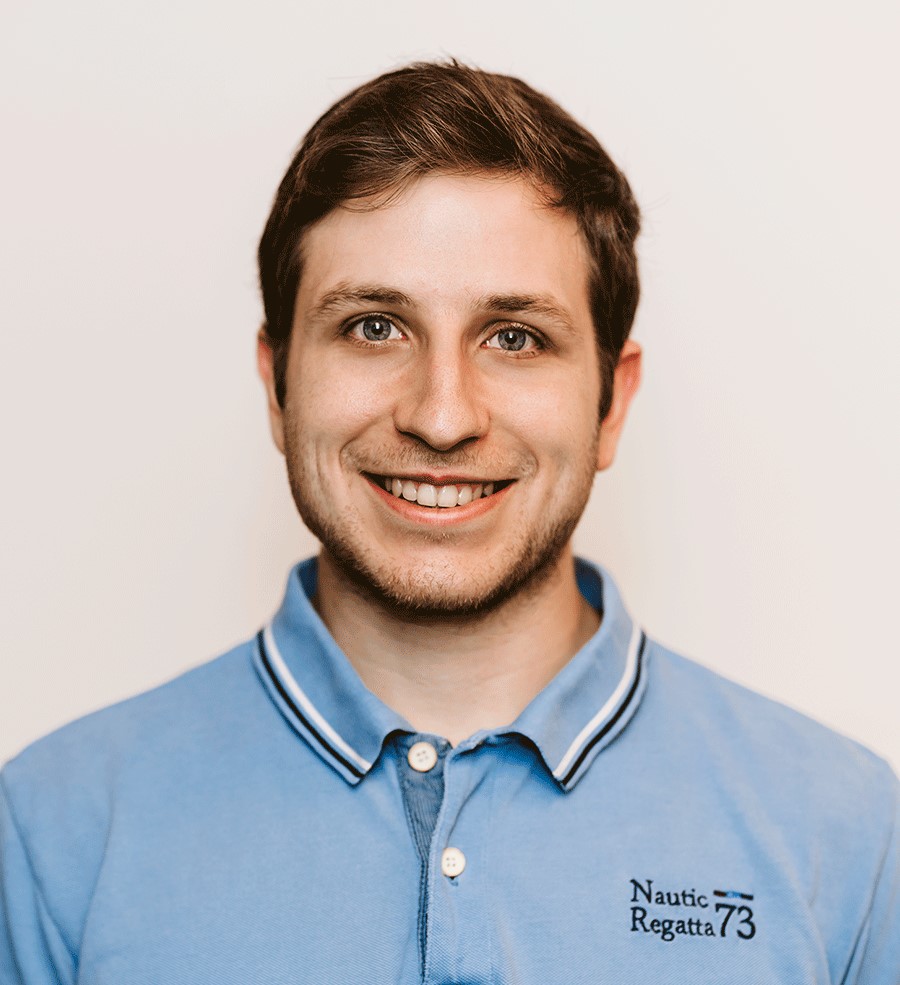}}]
{Andreu Cecilia} received the B.Eng. degree in industrial engineering, the double M.Sc. degree in
automatic control/industrial engineering and the PhD in automatic control from the
UPC, Barcelona, Spain, in 2017, 2020 and 2022, respectively. In 2022-2023,
he worked as a post-doctoral researcher at LAGEPP, Lyon, France. He is currently a post-doctoral researcher at UPC, Barcelona. His research interests include
observers, nonlinear systems and its application to energy systems and cyber-security.
\end{IEEEbiography}

\begin{IEEEbiography}
[{\includegraphics[width=1in,height=1.25in,clip,keepaspectratio]{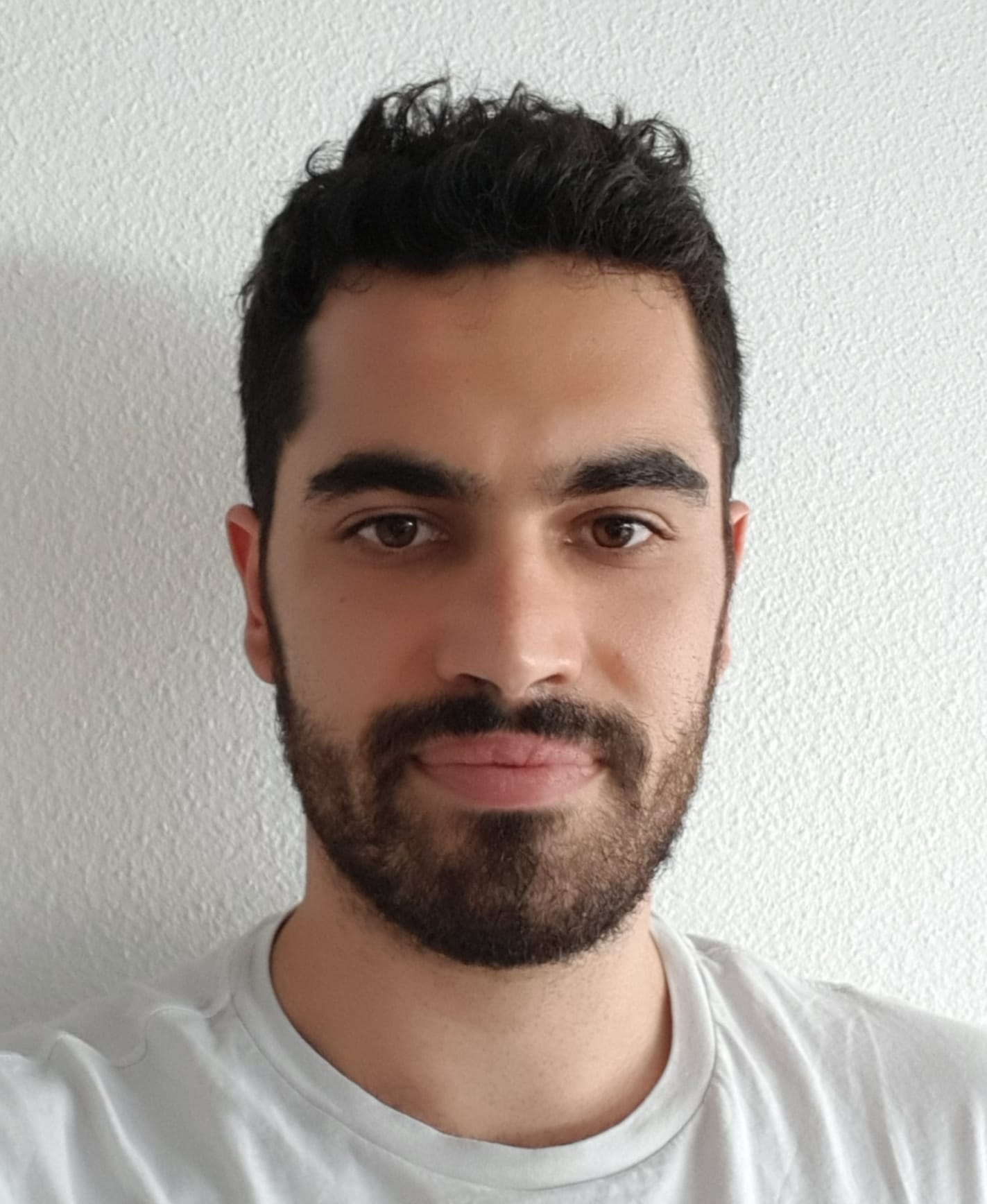}}]
{Samuele Zoboli} received the B.Sc. degree in electronics engineering from the University of Modena and Reggio Emilia, Italy, in 2016 and the M.Sc. degree in automation engineering from University of Bologna, Italy, in 2019. He received his
 Ph.D.  in Automatic Control in 2023 from University of Lyon 1, France, and he is now a post-doctoral researcher for University of Toulouse 3 at LAAS-CNRS, Toulouse, France. His research interests include stabilization of discrete-time nonlinear systems, multi-agent systems, reinforcement learning and control-applied artificial intelligence.
\end{IEEEbiography}

\begin{IEEEbiography}
[{\includegraphics[width=1in,height=1.25in,clip,keepaspectratio]{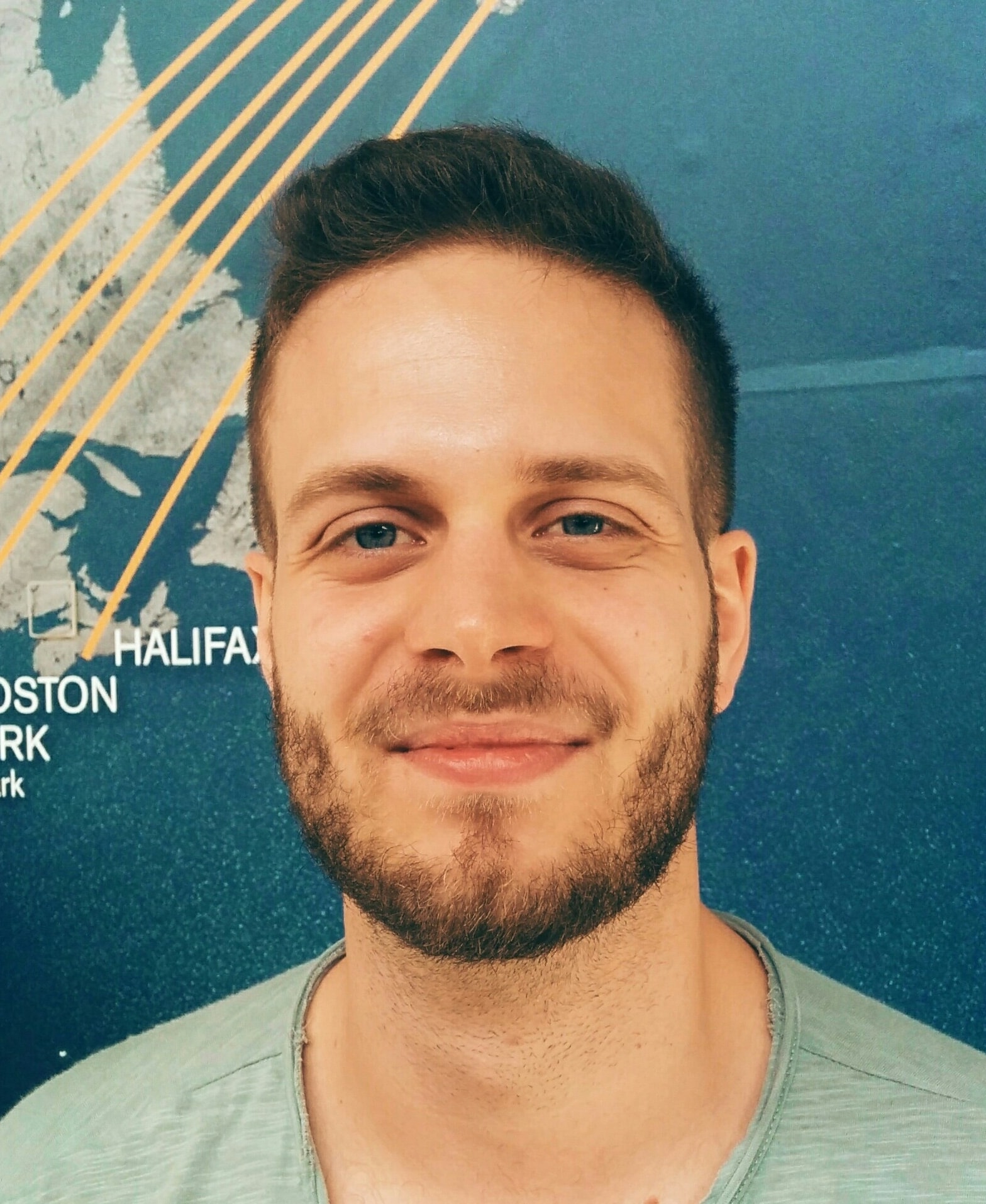}}]{Daniele Astolfi}
received the B.S. and M.S. degrees in
automation engineering from the University of Bologna,
Italy, in 2009 and 2012, respectively. He obtained a joint
Ph.D. degree in Control Theory from the University of
Bologna, Italy, and from Mines ParisTech, France, in 2016.
In 2016 and 2017, he has been a Research Assistant 
at the University of Lorraine (CRAN), Nancy, France.
Since 2018, he is a CNRS Researcher at 
LAGEPP, Lyon, France. 
His research interests include observer design, feedback
stabilization and output regulation  for nonlinear systems,
networked control systems, hybrid systems, and multi-agent systems.
 He serves as an associate 
editor of the IFAC journal Automatica.
He was a recipient of the 2016 Best Italian Ph.D. Thesis Award in Automatica given by Società Italiana Docenti e Ricercatori in Automatica (SIDRA, Italian Society of Professors and Researchers in Automation Engineering).
\end{IEEEbiography}

\begin{IEEEbiography}
[{\includegraphics[width=1in,height=1.25in,clip,keepaspectratio]{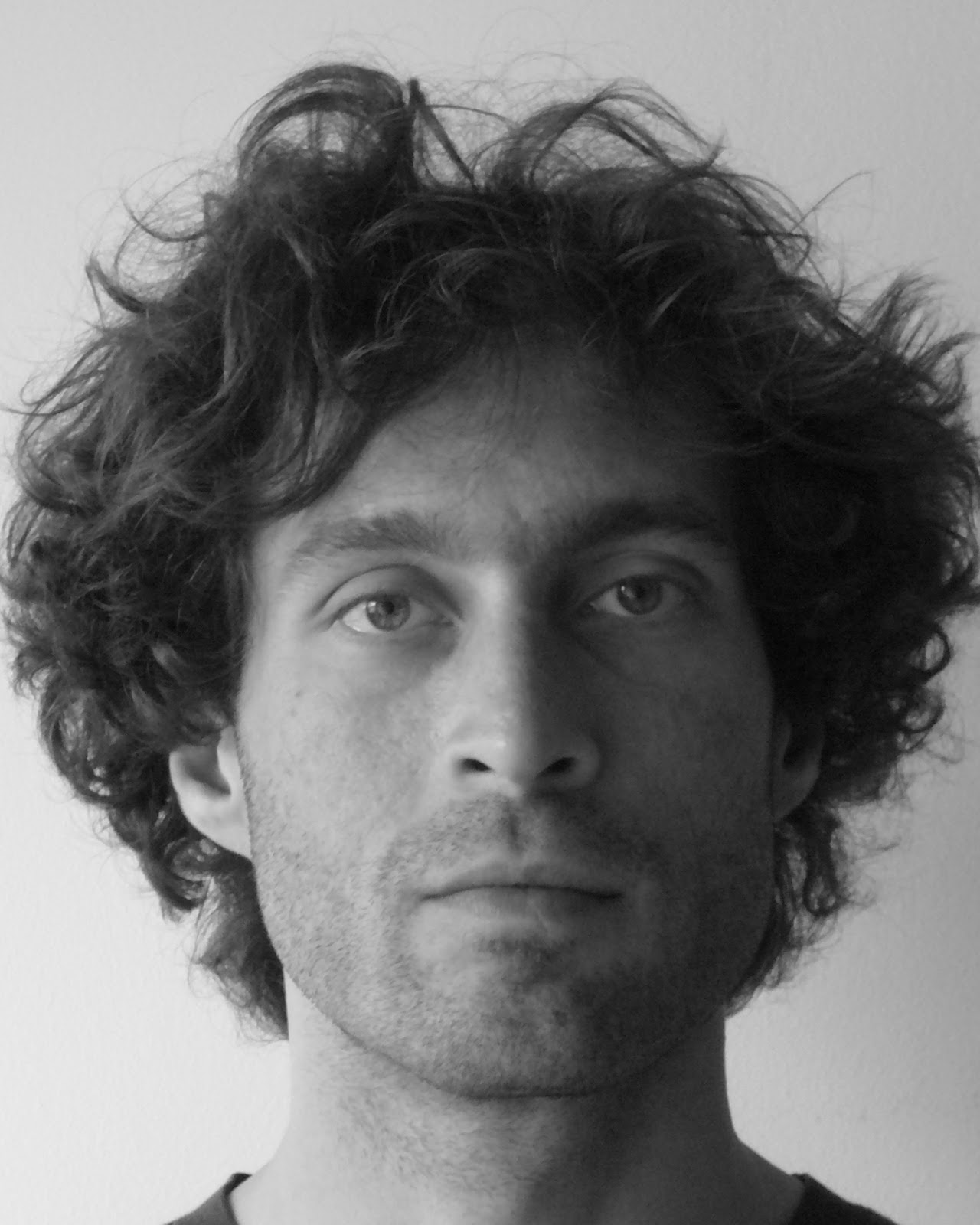}}]{Ulysse Serres}
received the Ph.D. degree in
mathematics from the Universite de Bourgogne,
Dijon, France, in 2006. Since 2009, he has been
with the Department of Electrical and Chemical
Engineering, Universite Claude Bernard Lyon 1,
Lyon, France, and with the LAGEPP Laboratory,
Lyon, France, where he is an Assistant Professor. His current research interests include
geometric optimal control, switching systems,
output feedback stabilization problems, and
(more recently) sub-Riemannian geometry.
\end{IEEEbiography}

\begin{IEEEbiography}
[{\includegraphics[width=1in,height=1.25in,clip,keepaspectratio]{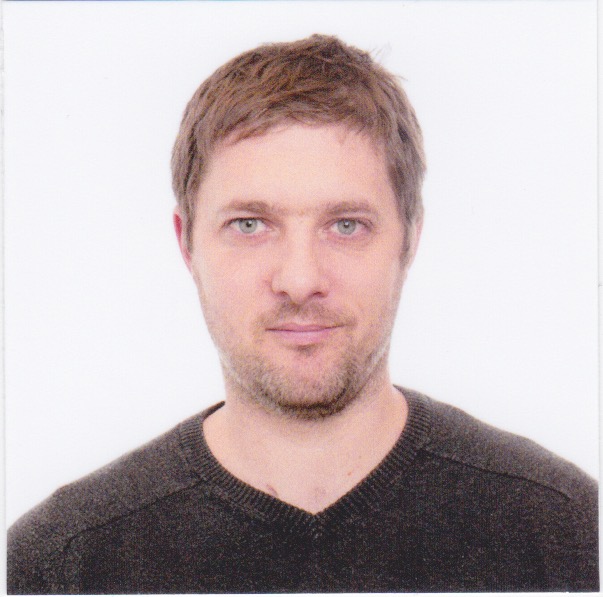}}]{Vincent Andrieu} is a Senior Reseacher at CNRS (Directeur de recherche). He graduated in applied mathematics from INSA de Rouen, France, in 2001. After working in ONERA (French aerospace research company), he obtained a PhD degree in control theory from Ecole
des Mines de Paris in 2005. In 2006, he had a research appointment at the Control and Power Group, Dept. EEE, Imperial College London. 
In 2008, he joined the CNRS-LAAS lab in Toulouse, France,
as a CNRS-charge de recherche. Since 2010,
 he has been working in LAGEPP-CNRS, Universite de
Lyon 1, France. In 2014, he joined the functional
analysis group from Bergische Universitat Wuppertal in Germany, for two sabbatical years. 
His main research interests are in the feedback stabilization of controlled dynamical nonlinear systems and state estimation problems. He is also interested in practical application of these theoretical problems, and especially in the field of aeronautics and chemical engineering. Since 2018 he is an associate editor of the IEEE Transactions on Automatic Control, 
and senior editor for System
\& Control Letters.
\end{IEEEbiography}

\end{document}